\documentclass[preprint,12pt]{elsarticle}
\pdfoutput=1

\usepackage{lineno}

\usepackage{upgreek}
\usepackage{amsmath} 
\usepackage{amssymb,mathtools}  
\usepackage{tabularx}
\usepackage{float}
\usepackage{stmaryrd}
\usepackage{algorithm}
\usepackage{algorithmic}
\usepackage{ulem}

\usepackage{caption,subcaption}
\usepackage{color}
\usepackage{multirow}

\newtheorem{assumption}{Assumption}
\newproof{proof}{Proof}

\usepackage{tabularx}

\def\Re{\mathbb{R}}
\def\R{\mathbb{R}}

\def\argmin{\mathop{\text{\rm arg\,min}}}

\def\Sec#1{Sec.~\ref{#1}}

\def\notes#1{\marginpar{\tiny #1}\typeout{Notes!
Notes!
Notes!
}}
\renewcommand{\notes}[1]{\typeout{notes!}}
\def\FRAC#1#2#3{\genfrac{}{}{}{#1}{#2}{#3}}
\def\half{{\mathchoice{\FRAC{1}{1}{2}}%
{\FRAC{2}{1}{2}}%
{\FRAC{3}{1}{2}}%
{\FRAC{4}{1}{2}}}}

\def\Re{\field{R}}

\def\Sec#1{Sec.~\ref{#1}}

\def\transpose{{\hbox{\rm\tiny T}}}

\def\Sec#1{Sec~\ref{#1}}

\def\E{{\sf E}}
\def\Expect{{\sf E}}

\def\Prob{{\sf P}}

\def\Expect{{\sf E}}

\def\wham#1{\smallbreak\noindent\textbf{\textit{#1}}}

\def\R{\mathbb{R}}

\def\Sec#1{Sec.~\ref{#1}}

\newcommand{\expect}{ {\sf E} }

\newtheorem{lemma}{Lemma}

\newtheorem{proposition}{Proposition}

\def\beq{\begin{eqnarray}} 
\def\bc{\begin{center}} 
\def\be{\begin{enumerate}}
\def\bi{\begin{itemize}} 
\def\bs{\begin{small}}
\def\bS{\begin{slide}}
\def\ec{\end{center}} 
\def\ee{\end{enumerate}}
\def\ei{\end{itemize}}
\def\es{\end{small}}
\def\eS{\end{slide}}
\def\eeq{\end{eqnarray}}

\newcommand{\trace}{\text{Tr}}

\newcommand{\ud}{\,\mathrm{d}}

\def\Re{\mathbb{R}}
\def\E{{\sf E}}

\def\argmin{\mathop{\text{\rm arg\,min}}}

\def\Sec#1{Sec.~\ref{#1}}

\def\Prop#1{Prop.~\ref{#1}}

\def\Expect{{\sf E}}

\renewcommand{\Re}{\mathbb{R}}

\def\Prob{{\sf P}}
\def\FRAC#1#2#3{\genfrac{}{}{}{#1}{#2}{#3}}

\newcommand{\Ricc}{\text{Ricc}}

\newcommand{\Ybar}{\bar{X}}
\newcommand{\Sbar}{\bar{S}}
\newcommand{\SN}{S^{(N)}}

\newcommand{\normal}{\mathcal{N}}
\newcommand{\hP}{h}
\newcommand{\vP}{\mathcal{V}}

\newcommand{\id}{\mathcal{I}}

\newcommand{\RU}{R}
\newcommand{\backward}[1]{\overset{\shortleftarrow}{#1}}

\newcommand{\stepsize}{\Delta t}

\newcommand{\iid}{\stackrel{\text{i.i.d}}{\sim}}

\newcommand{\Ybarbar}{{{Y}}}
\newcommand{\etabar}{{\eta}}

\newcommand{\tp}{^{\top}}

\def\Qrob{{\sf Q}}

\makeatletter
\newcommand\gobblepars{%
    \@ifnextchar\par%
        {\expandafter\gobblepars\@gobble}%
{}}
\makeatother

\def\wham#1{\smallbreak\pagebreak[3]%
	\noindent\textbf{#1}\ \ \gobblepars}

\def\Prop#1{Prop.~\ref{#1}}

\def\Ybar{{X}}
\def\clA{\mathcal{U}}

\def\transpose{{\intercal}}

\def\fee{\upphi}

\renewcommand{\em}{\it}

\journal{Systems and Control Letters}

\begin{document}

\begin{frontmatter}



\title{Controlled Interacting Particle Algorithms for  Simulation-based Reinforcement Learning}


\author[a1]{Anant A. Joshi}

\address[a1]{Coordinated Science Laboratory and the Department of Mechanical Science and Engineering at the University of Illinois at Urbana-Champaign
            }


\author[a3]{Amirhossein Taghvaei}

\address[a3]{William E. Boeing Department of Aeronautics and Astronautics at University of Washington Seattle
            }
    
\author[a1]{Prashant G. Mehta}    
            
\author[a4]{Sean P. Meyn}
\address[a4]{Department of Electrical and Computer
  Engineering at University of Florida Gainesville
            }

\vspace*{0.5in}

\begin{abstract}

This paper is concerned with optimal control problems for control systems in continuous time, and interacting particle system methods designed to construct approximate  control solutions.   
Particular attention is given to the linear quadratic (LQ) control problem.  
There is a growing interest in re-visiting this classical problem, in part due to the successes of reinforcement learning (RL).  The main question of this body of research (and also of our paper) is to  approximate the optimal control law {\em without} explicitly solving the Riccati equation.     
A novel simulation-based algorithm, namely a dual ensemble Kalman
filter (EnKF), is introduced.  The algorithm is used to obtain
formulae for optimal control, expressed entirely in terms of the EnKF
particles. An extension to the nonlinear case is also presented.   The
theoretical results and algorithms are illustrated with numerical
experiments.  


\end{abstract}
 


\noindent 
\large {\it Dedicated to the memory of Ari Arapostathis}

\vspace*{0.5in}

%
%
%
%

\end{frontmatter}




\section{Introduction}
\label{sec:intro}


The field of reinforcement learning (RL) is concerned with optimal
control, to design a policy for a dynamical system that minimizes some
performance criterion.  All of the standard choices are treated in the
literature:  discounted cost, finite time-horizon, and average cost.
What makes the RL paradigm so different from optimal control as
formalized by Bellman and Pontryagin in the 1950s 
is that in RL the system identification step is usually avoided.
Instead, the optimal policy is approximated based on input-output
measurements.  
 
There are two standard approaches to obtain an algorithm for this
purpose:  (i)  critic methods, in which a value function is
approximated within a parameterized family, and the policy is obtained
as a functional of the approximation, and (ii) actor methods in which
a parameterized family of policies is given, and the algorithm is
designed to obtain the best policy within this family.

In popular media, RL is often described as an ``agent'' that learns an
approximately optimal policy based on interactions with the
environment.  Important examples of this ideal include advertising,
where there is no scarcity of real-time data.  In the vast majority of
applications we are not so fortunate, which is why successful
implementation usually requires simulation of the physical system for
the purposes of training.   For example, DeepMind's success story with
Go and Chess required weeks of simulation for training on a massive
collection of
super-computers~\cite{mu-zero_2019MasteringAG}.

This paper focuses on model-based RL in which 
the model is available only in the form of a simulator.  
The proposed approach is novel,  drawing
on mean-field techniques similar to those appearing in 
state estimation (data assimilation) in high dimension.  
It is likely that the concepts will lead to new approaches for online
RL---see directions for future research in the conclusions.    
 
We consider the finite-horizon optimal control problem,
\begin{subequations}\label{eq:nonlinear_opt_control_problem}
\begin{align}
\min_{u} \quad J(u) &= 
                                    \int_0^T 
                                    \left(\half |c(x_t)|^2 + \half u_t^\top \RU u_t \right)    \ud t  + g(x_T)
         \\
\text{subject to:} \quad \dot{x}_t &= a(x_t) + b(x_t) u_t , \; x_0=x\label{eq:nonlinear:model-dynamics}
\end{align} 
\end{subequations}
where  $x_t\in\Re^d$ is the state at time $t$, and $u=\{u_t\in\Re^m:
0\leq t\leq T\}$ is the control input.   The functions
 $a(\cdot)$, $b(\cdot)$, $c(\cdot)$, $g(\cdot)$ are
continuously differentiable ($C^1$), and the control penalty matrix positive definite, $\RU\succ 0$.  

In the 
linear quadratic (LQ) setting the model is linear ($a(x) = A
x$ and $b(x) = B$) and the cost function is quadratic ($c(x) = C x$
and $g(x) = x^\transpose P_T x$).  
The infinite-time horizon
($T=\infty$) case is referred to as the linear quadratic regulator
(LQR) problem.  Although it is a historical problem, LQR has been the
subject of recent research interest in the control community.
The goals of this research are much like ours: design simulations for
the purposes of learning the optimal control policy.        

The proposed solution involves construction of $N$ stochastic processes $\{Y_t^i\in\Re^d:   0\leq t\leq T, 1\leq i\leq N\}$ where the $i$-th particle evolves according to a stochastic differential equation (SDE) of the form,
\begin{equation}\label{eq:ybar_intro}
\ud Y^i_t =  \underbrace{a(Y_t^i) \ud t + b  (Y_t^i) \ud
  v_t^i}_{i-\text{th copy of
    model}~\eqref{eq:nonlinear:model-dynamics}} \; +\;
\clA_t^i \ud t,\;\; 0\leq t\leq T
\end{equation}
where the input $v^i=\{v_t^i\in\Re^m:0\leq t\leq T\}$ and the \textit{data assimilation
  process} $\clA^i=\{\clA_t^i\in\Re^d:0\leq t\leq T\}$ is obtained as part of the RL design.
The goal is to  design these processes so that  the  empirical
distribution of the $N$ particles at time $t$ approximates a smooth density $p_t$ (for the $N=\infty$ mean-field limit), encoding the
optimal policy as follows:
\begin{equation}
\fee_t^*(x) =   \RU^{-1}b^\top (x) \nabla \log {p}_t(x)
\label{e:policy}
\end{equation}
where $\nabla$ denotes the gradient operator.
In the
infinite-horizon case, a stationary policy is obtained by letting
$T\to\infty$. 
  

We make the following assumption:
 
 \begin{assumption}
 \label{ass:Ass1}
\begin{enumerate}
\item Functions $f(x,\alpha)=a(x) + b(x) \alpha$ and $c(x)$ are
  available in the form of an oracle (which allows function
  evaluation at any state action pair $(x,\alpha)\in\Re^d\times \Re^m$).
\item Matrices $R$ and $P_T$ are available.  Both of these matrices
  are strictly positive-definite.
\item Simulator is available to simulate (2).  In particular, this
  requires an ability to add additional inputs outside the control
  channel. 
\end{enumerate}
\end{assumption}

Part 1 of the assumption is standard for any RL algorithm.  Part 2 is
not too restrictive for the following reasons: In physical systems,
one typically is able to assess relative costs for different control
inputs (actuators).  For the LQR problem, under certain technical
conditions, the optimal policy is stationary and does not depend upon
the choice of $P_T$.  A possibility is to take $R$ and $P_T$ to be
identity matrices of appropriate dimensions.  The main restriction
comes from part 3 of the assumption.

A motivation comes from data assimilation applications such as weather prediction and geosciences where
Assumption~\ref{ass:Ass1} is standard.  The ensemble Kalman Filter
(EnKF) is a particle system method which serves as a workhorse in
these applications  \cite{evensen2006,reich2015probabilistic}. The
computational complexity  of the EnKF is $\mathcal{O}(N d)$ where $d$
is dimension, and $N$ is the number of particles,   with $N\ll d$
typical in these applications.

Part of the tremendous success of the EnKF in these domains is that it works directly with the
simulator. Multiple copies are run in a Monte-Carlo manner where the
data assimilation process is used to assimilate the most recent
measurement.

\textit{The goal of the research summarized here is to create approximation techniques with similar success for applications in control.
}   



\subsection{Contributions of this paper}

In order to elucidate these new ideas as clearly as possible, the main
focus of this paper is on the linear quadratic (LQ) problem. 
This also allows us to highlight and contrast our work with recent
developments.   
The
algorithm~\eqref{eq:ybar_intro} for the LQ problem is presented first
in \Sec{sec:LQ} before describing its nonlinear extension in \Sec{sec:nonlinear}.  The
details of the original contributions of our work are as follows:

\smallskip

\noindent \textbf{1.} For the LQ problem, the proposed algorithm is an ensemble Kalman
  filter (EnKF) referred to here as the {\em dual EnKF}.  The
  mean-field limit ($N=\infty$) of the dual EnKF is shown to be exact
  (\Prop{prop:Y-exactness}).
For the finite-$N$ algorithm, an error
bound on the approximation is obtained  under 
additional assumptions on the model (see~\eqref{eq:error_formula}).  An extensive discussion is
included in \Sec{sec:LQ-compare} to situate
the algorithm in the RL landscape.  In particular, it is shown that
(i) the process $v^i$ implements the exploration step of RL whereby the cheap
control directions are explored more; and (ii) the process $\clA^i$
implements the value iteration step of RL. 

\smallskip

\noindent \textbf{2.} For the nonlinear
  problem~\eqref{eq:nonlinear_opt_control_problem}, a dual algorithm
  is presented to approximate the Hamilton-Jacobi-Bellman (HJB)
  equation.  The algorithm requires a solution of a (linear) Poisson
  equation that is far more easily approximated.  It is shown that the
  dual EnKF algorithm for the LQ problem is a special case in which
  the Poisson equation admits an analytical solution.  

\smallskip

\noindent \textbf{3.} A numerical comparison of the dual EnKF algorithm against the
  state-of-the-art is described for benchmark examples.  It is shown
  that the proposed algorithm can be up to two orders
  of magnitude more computationally efficient (Fig.~\ref{fig:comp-to-lit}).
  Scalings with respect to both the number of particles $N$ and the
  problem dimension $d$ are numerically illustrated and compared with
  analytical bounds (Fig.~\ref{fig:mse}).  

\subsection{Literature review}

There are three areas of prior work that are related to the
subject of this paper: (i) RL algorithms for the LQR problem; (ii)
EnKF and related control-type algorithms for data assimilation; and
(iii) duality theory between optimal control and estimation.  

\smallskip

\noindent \textbf{(i) RL for LQR:} The LQ
problem   
has a rich and storied history in modern control
theory going back to the original work on the
subject~\cite{kalman-1960-o}.  Its solution requires solving a Riccati
equation -- the differential Riccati equation (DRE) in finite-horizon
settings or the algebraic Riccati equation (ARE) in the
infinite-horizon setting.   
There is a large body
of literature devoted to the analytical study of the Riccati
equations~\cite{lancaster1995algebraic} 
and
specialized numerical techniques have been developed to efficiently
compute the solution~\cite{laub1991invariant}.

There are two issues which makes the LQ and related problems a topic of
recent research interest: (i) In high-dimensions, the
matrix-valued nature of the DRE or ARE means that any algorithm is
$\mathcal{O}(d^2)$ in the dimension $d$ of the state-space, and (ii) the
model parameters may not be explicitly available to write down the DRE
(or the ARE) let alone solve it.  The latter is a concern, e.g., when
the model exists only in the form of a black-box numerical simulator.  

These two issues have motivated the recent research on the
infinite-horizon linear quadratic regulator (LQR) problem~\cite{fazel_global_2018,tu_gap_2019,dean_sample_2020,
malik_derivative-free_2020,mihailo-2021-tac}.  
Of particular interest are policy gradient type algorithms that seek to bypass solving an ARE by directly searching over
the space of stabilizing gain matrices. 
 The algorithms are of
iterative type where each iteration requires a policy evaluation step
(using $N$ simulations much like~\eqref{eq:ybar_intro}).  This step is used to estimate a gradient
which is then used to obtain a improved policy based on a gradient
descent procedure. 
Global convergence rate
estimates are established for both
discrete-time~\cite{fazel_global_2018,mohammadi_linear_2021}
and continuous-time~\cite{mihailo-2021-tac} settings of the LQR problem.  Extensions to
the $H_\infty$ regularized LQR~\cite{zhang_stability_2020} and Markov jump linear
systems~\cite{jansch2020convergence} have also been carried out.  In
the recent thesis \cite[Chapter 4]{kqz-thesis}, finite horizon
extensions are considered under additional assumptions.

Additional comparison with this prior work appears in
\Sec{sec:LQ-compare} and numerical comparison is in \Sec{sec:numerics}.

\smallskip

\noindent \textbf{(ii) EnKF for data assimilation:} Although novel for RL, the proposed algorithms are inspired by the 
data assimilation (nonlinear filtering)
literature~\cite{reich2015probabilistic}.  
During the past decade, a key breakthrough in the data assimilation
theory and its applications is the design of controlled interactions between
particles (such as $\clA_t^i$ in~\eqref{eq:ybar_intro}) to approximate the
solution of the nonlinear filtering problem; c.f.,~\cite{TaghvaeiCSM21} and 
references therein.  Such an approach is in contrast to
the traditional 
importance sampling type approaches which suffer from 
issues such as particle degeneracy~\cite{surace_SIAM_Review}.  Two
well known examples of the controlled interacting particle systems are
the ensemble Kalman filter (EnKF) and the feedback particle filter (FPF).  The EnKF is an exact algorithm for the linear
Gaussian filtering problem~\cite{reich11,bergemann2012ensemble} while the FPF is an exact
algorithm for the nonlinear non-Gaussian
case~\cite{taoyang_TAC12}. 
The two major algorithmic contributions of this work represent the optimal
control (dual) counterparts of the EnKF and the
FPF.  

Notably, EnKF is a
workhorse in data assimilation applications such as weather prediction
where models are simulation-based and
high-dimensional \cite{evensen2006,reich2015probabilistic}.  As noted
above, these two
issues have also motivated much of recent work on the LQ problem in
the control community.  

\smallskip

\noindent \textbf{(iii) Duality:} The formula~\eqref{e:policy} for the
optimal policy is a consequence of the so-called log transformation.  The
transformation relates the value function of an optimal control
problem to the posterior distribution of the dual optimal estimation
problem~\cite{fleming-1982}.
Duality is an old subject~\cite[Chapter 7.5]{astrom},\cite[Chapter 15]{kailath2000linear}.  In recent years, there has
been renewed interest in duality for algorithm design.  In much of the classical
literature on the subject, duality was used to obtain optimal control
algorithms for solving estimation problems~\cite{mortensen-1968}.  Although
it remains an important theme~\cite{kim2020optimal},
some of the more recent work has been in the opposite direction: to
solve optimal control problems by using sampling
techniques~\cite{levine-2018}.  Our work fits within this latter body
of work.  
 
A salient aspect of this paper is a detailed comparison with
literature appearing in each of the three main sections after
technical details have been presented. 



\subsection{Paper outline}

The outline of the remainder of this paper is as follows.  The LQ
optimal control  problem and its dual EnKF solution is described
in~\Sec{sec:LQ}.  The nonlinear extension and its connection to
duality appears
in~\Sec{sec:nonlinear}.  The algorithms are illustrated with numerical
examples in~\Sec{sec:numerics}.  
The proofs appear in the
Appendix. 
 
\smallskip

{\bf \noindent Notation:}
${\cal N}(m,\Sigma)$ is a Gaussian probability
distribution with mean $m$ and covariance $\Sigma$.  The notation  $\Sigma \succ 0$ is used when the matrix $\Sigma$ is positive definite. The $n \times n$ identity matrix is denoted $\id_n$. The trace of a matrix is denoted by $\trace(\cdot)$.
For a smooth function $f:\Re^d\to \Re$, $\nabla f(x) = [\frac{\partial f}{\partial x_1},\ldots,\frac{\partial f}{\partial x_d} ]^\top$ denotes the gradient of $f$, and $\nabla^2f(x) = [\frac{\partial^2f}{\partial x_m \partial x_n}(x)]_{n,m=1}^d$ denotes the Hessian matrix. For a smooth vector-field $v:\Re^d\to\Re^d$, $\nabla \cdot v(x) = \sum_{n=1}^d \frac{\partial v_n}{\partial x_n}(x)$ denotes the divergence.  And for a smooth tensor $D:\Re^d \to \Re^{d\times d}$, $\nabla \cdot D(x)$ is a vector field  whose $m$-th component is  $ \sum_{n=1}^d \frac{\partial D_{mn}}{\partial x_n}(x)$, and $\nabla^2 \cdot D(x) = \sum_{n,m=1}^d \frac{\partial^2 D_{mn}}{\partial x_n \partial x_m}(x)$.

 \section{LQ problem}
\label{sec:LQ}


\begin{subequations} 
\label{eq:LQ}

The finite-horizon linear quadratic (LQ) optimal control problem
is a special case of~\eqref{eq:nonlinear_opt_control_problem} as follows:
\begin{align}
\min_{u} \quad J(u) &= 
                                    \int_0^T \half |C x_t|^2 + \half
u_t^\top \RU u_t \ud t + \half x_T^\top
P_T x_T  \label{eq:LQ:model-objective} \\
\text{subject to:} \quad \dot{x}_t &= A x_t + B u_t, \quad x_0=x \label{eq:LQ:model-dynamics}
\end{align} 
It is assumed that $(A,B)$ is controllable, $(A,C)$ is observable,
and the matrices $P_T,\;\RU \succ 0$.  The $T=\infty$ limit is referred to as the linear quadratic regulator (LQR)
problem.   

\end{subequations}

It is well known
that the optimal control $u_t = \fee_t(x_t)$ where the optimal policy is linear
\[
\fee_t(x)=K_t x \quad \text{where}\quad K_t = - \RU^{-1}B^\top P_t
\]
is the optimal gain matrix and $P_t$ is a solution of the backward (in time) DRE
\begin{equation}\label{eq:Ricatti}
-\frac{\ud}{\ud t} P_t =A^\top P_t + P_t A + C^\top C  - P_tB  \RU^{-1}B^\top P_t,
\quad P_T \;\text{(given)}
\end{equation}
The ARE is obtained by setting the left-hand side to $0$.  As $T\to\infty$, 
for each fixed time $t$, $P_t \to P^{\infty}$, exponentially fast
\cite[Remark 2.1]{ocone1996}, where $P^{\infty} \succ 0$ is the unique
positive-definite solution of the ARE, and therefore the optimal gain converges,
$K_t \to K^{\infty}:=- \RU^{-1}B^\top P^{\infty}$.   Approximation of the LQR gain $K^{\infty}$ is
a goal in recent  RL
research~\cite{fazel_global_2018,mihailo-2021-tac}.          

\subsection{Main contribution: Dual EnKF algorithm}


Under the assumptions of this paper, $P_t\succ 0$ for $0\leq t\le T$ whenever $P_T\succ
0$~\cite[Sec.~24]{brockett2015finite}.
Set $S_t = P_t^{-1}$.  It is readily verified that $S_t$ also solves a backward DRE  
\begin{equation}\label{eq:backward-DRE}
\frac{\ud}{\ud t} S_t = AS_t + S_t A^\top - B  \RU^{-1}B^\top +
S_tC^\top C S_t,\quad S_T=P_T^{-1}
\end{equation}     

Our objective is to approximate $S_t$ using simulations. The proposed construction proceeds
in two steps: (i) definition of an exact mean-field process and (ii) its finite-$N$
approximation.


\medskip

\noindent \textbf{Step 1. Mean-field process:} 
Define $\Ybarbar=\{\Ybarbar_t\in\Re^d:
0\leq t\leq T\}$ as a solution of the
following backward (in time) stochastic differential
equation (SDE):
\begin{align}
\ud \Ybarbar_t & = A \Ybarbar_t \ud t   +  B\ud  \backward{\etabar}_t
                 +\half \Sbar_tC^\top (C\Ybarbar_t +
  C \bar{n}_t) \ud t,  \;0\leq t<T\nonumber \\ \Ybarbar_T & \stackrel{\text{d}}{=} \mathcal
  N(0,S_T) \label{eq:Ybar} 
\end{align}
where ${\etabar}=\{{\etabar}_t \in \Re^m:0\leq t\leq T\}$ is a Wiener
process (W.P.) 
with covariance matrix $\RU^{-1}$, and 
\begin{align}
\bar{n}_t:=\expect[\Ybarbar_t], \quad \Sbar_t
:=\Expect[(\Ybarbar_t-\bar{n}_t)(\Ybarbar_t-\bar{n}_t)^\top] \label{eq:Sbar}
\end{align}
The process $Y$ is an example of a
mean-field process because its evolution depends upon the statistics
($\bar{n}_t,\bar{S}_t$) of the process.  An SDE of this type is called a McKean-Vlasov SDE.  The meaning of the backward
arrow on $\ud  \backward{\etabar}$ in~\eqref{eq:Ybar} is that the SDE
is simulated backward in time starting from the terminal condition
specified at time $t=T$.  The reader is referred
to~\cite[Sec. 4.2]{nualart1988stochastic} for the definition of the
backward It\^{o}-integral. 

The mean-field process is useful because of the following proposition
whose proof is included in  \ref{app:prop-y-exact}.
\medskip

\begin{proposition}\label{prop:Y-exactness}
The solution to 
the SDE~\eqref{eq:Ybar}
is a Gaussian stochastic process, in which the mean and covariance of $\Ybarbar_t$ are given by
\[
\bar{n}_t = 0,\quad \bar{S}_t=S_t, \quad \;
0\leq t\leq T
\]
Consequently, $\Ybar_t  := \Sbar_t^{-1} (\Ybarbar_t-\bar{n}_t)$ is also
a Gaussian satisfying 
\[
{\sf E}(\Ybar_t) = 0,\quad {\sf E}(\Ybar_t \Ybar_t^\top) = P_t, \quad\;
0\leq t\leq T
\]
\end{proposition}

The significance of~\Prop{prop:Y-exactness} is that the optimal control
policy $\fee_t(\cdot)$ can now be obtained in terms of the statistics of the random variable
$\Ybar_t$.  Specifically, we have the following two cases: 
\begin{enumerate}
\item[(i)] If the matrix $B$ is explicitly known then the optimal gain matrix 
\[
K_t =  - \RU^{-1}B^\top {\sf E}(\Ybar_t
\Ybar_t^\top)
\]
\item[(ii)] If $B$ is unknown, define the Hamiltonian (the
continuous-time counterpart of the Q-function~\cite{mehta2009q}):
\[ H(x,\alpha,t) := \underbrace{\half
|Cx|^2 + \half \alpha^\top R \alpha}_{\text{cost function}} +
  x^\top {\sf E}(\Ybar_t \Ybar_t^\top) \underbrace{(Ax + B\alpha)}_{\text{model}~\eqref{eq:LQ:model-dynamics}} 
\]
 from which the
  optimal control law is obtained as
\[
\fee_t(x) = \argmin_{\alpha\in\Re^m} H(x,\alpha,t)
\] 
by recalling the minimum principle, which states that the optimal control is the unique minimizer of the
Hamiltonian. It is noted that the Hamiltonian $H(x,\alpha,t)$ is in the form of an oracle because $(Ax
    + B\alpha)$ is the right-hand side of the simulation
    model~\eqref{eq:LQ:model-dynamics}.
\end{enumerate}


\medskip

\noindent \textbf{Step 2. Finite-$N$ approximation:} 
The mean-field process is empirically approximated by simulating a
system of controlled interacting particles 
according to 
\begin{align}
\ud {Y}^i_t & = \underbrace{A {Y}^i_t \ud t + B\ud \backward{\eta}^i_t}_{\text{i-th copy of model}~\eqref{eq:LQ:model-dynamics}} +
                   \underbrace{S^{(N)}_tC^\top \left(\frac{C{Y}^i_t+Cn^{(N)}_t}{2} \right)}_{\text{data assimilation process}}\ud
  t, \label{eq:EnKF-Yi} \\
{Y}^i_T & \stackrel{\text{i.i.d}}{\sim} \mathcal
N(0,P_T^{-1}),\quad 1\leq i\leq N \nonumber
\end{align}
${\eta}^i$ is an i.i.d copy of ${\etabar}$, $n^{(N)}_t = N^{-1}\sum_{i=1}^N
{Y}^i_t$, and 
\begin{align*}\SN_t= \frac{1}{N-1}\sum_{i=1}^N
  ({Y}^i_t-n^{(N)}_t)({Y}^i_t-n^{(N)}_t)^\top \label{eq:SbarN}
\end{align*}
The data assimilation process has a linear feedback control structure
where $S^{(N)}_tC^\top$ is the Kalman gain matrix and $\frac{1}{2} (C
Y_t^i + Cn_t^{(N)})$ is the state feedback term similar to the error in the
FPF~\cite{yang-2016}.  The process serves
  to couple the particles.  Without it, the particles are
  independent of each other.

The finite-$N$ system~\eqref{eq:EnKF-Yi} is referred to as the {\em
  dual EnKF}.  

\medskip

\noindent \textbf{Optimal control:} Set $
\Ybar^i_t  = (\SN_t)^{-1} ({Y}^i_t  - n^{(N)}_t)$.  
There are two cases as before: 
\begin{enumerate}
\item[(i)] If
the matrix $B$ is explicitly known then
\begin{equation}
K_t^{(N)}  = - \frac{1}{N-1} \sum_{i=1}^N \RU^{-1}
  (B^\top  \Ybar^i_{t} )(\Ybar^i_{t})^\top 
\label{eq:KtN}
\end{equation}
\item[(ii)] If $B$ is unknown, define the Hamiltonian 
\begin{align*}
 H^{(N)}(x,\alpha,t) :=
 \underbrace{\half
|Cx|^2 + \half \alpha^\top R \alpha}_{\text{cost function}} +
 \frac{1}{N-1} \sum_{i=1}^N  (x^\top  X^i_{t} )(X^i_{t})^\top \underbrace{(Ax + B\alpha)}_{\text{model}~\eqref{eq:LQ:model-dynamics}} 
\end{align*}
from which the
  optimal control policy is approximated as
\[
\fee_t^{(N)}(x) = \argmin_{a\in\Re^m} H^{(N)}(x,a,t)
\]
There are several zeroth-order approaches to solve the minimization problem, e.g., by
constructing 2-point estimators for the
gradient.  Since the
objective function is quadratic and the matrix $R$ is known, $m$ queries
of $H^{(N)}(x,\cdot,t)$ are sufficient to compute $\fee_t^{
  (N)}(x)$. 

\end{enumerate}

The overall algorithm including its numerical approximation appears in
\ref{app:alg-NL-EnKF}. 





\subsection{Remarks}

The following remarks are included to help provide an intuitive
explanation to various aspects of the dual EnKF.   

	\wham{1. Representation.}
 In designing any RL algorithm, the
first issue is representation of the unknown
value function ($P_t$ in the linear case).  Our novel idea is to represent $P_t$ is in terms of statistics
(variance) of the particles.  Such a representation is fundamentally
distinct from representing the value function, or its proxies, such as
the Q function, within a parameterized class of functions. 
 

	\wham{2. Value iteration.}

The algorithm is entirely simulation based: $N$ copies of the
model~\eqref{eq:LQ:model-dynamics} are simulated in parallel where the
terms on the right hand-side of~\eqref{eq:EnKF-Yi} have the following
intuitive interpretations:

\begin{enumerate}
\item {\em Dynamics:} The first term ``$A {Y}^i_t \ud t$'' on the right-hand side of~\eqref{eq:EnKF-Yi} is
  simply a copy of uncontrolled dynamics in the
  model~\eqref{eq:LQ:model-dynamics}.

\item {\em Control:} The second term is 
the control input $U$ for the $i$-th particle, specified as a white noise process with covariance $\RU^{-1}$.  One may interpret this as an approach to exploration whereby cheaper control
  directions are explored more.   
\end{enumerate}


While there are similarities with traditional approaches to RL, the novelty comes
from the data assimilation process that represents an original contribution.

	\wham{3. Arrow of time.}

The particles are simulated
backward -- from terminal time $t=T$ to initial time $t=0$.  This is 
consistent with the dynamic programming (DP) equation which
also proceeds backward in time.

\vspace*{-0.02in}
\subsection{Convergence and error analysis}\label{sec:conv_anal}
\vspace*{-0.02in}

The mean-field process~\eqref{eq:Ybar} represents the mean-field
limit of the finite-$N$ system~\eqref {eq:EnKF-Yi}, as the number of
particles $N\to \infty$. The convergence analysis is a challenging problem
but impressive progress has been made in some groundbreaking work
appearing in
recent years~\cite{bishop2018stability,delmoral-2020}.  In \ref{app:error-proof}, under
certain additional assumptions on system matrices, the following
error bound is derived:
\begin{equation}\label{eq:error_formula}
\Expect[\|S^{(N)}_{t}-\bar{S}_t\|_F] \leq
\frac{C_1}{\sqrt{N}} + C_2e^{-2\lambda (T-t)} \Expect[\|\SN_T-\Sbar_T\|_F],
\end{equation}
where $C_1,C_2$ are  time-independent positive constants, and $||\cdot||_F$ denotes Frobenius norm for matrices.
The proof largely follows the
techniques developed in~\cite{delmoral-2020}. 



\subsection{Comparison to literature}
\label{sec:LQ-compare}

\wham{Function approximation:} 
Classical RL algorithms for the LQR problem are based on a linear
function approximation, using quadratic basis functions, of the value
function or the Q-function \cite{bradtke-1992,bradtke-1994, lewis-2011-adprl}. 
The basic idea is to run the system for a time horizon $T$, and successively update an estimate of the parameters based on new data collected, using a least-squares approximation. 
Convergence guarantees typically require (i) a persistence of
excitation condition, see e.g. \cite[Equation (9)]{bradtke-1994}, \cite[Remark 3, Page 173]{lewis-2015-tc} and (ii) use of the on-policy methods whereby the
parameters are learned for a given fixed policy (which is subsequently
improved), see e.g. \cite[Page 299]{bradtke-1992}.  
For the deterministic LQR problem, the persistence of
excitation condition is difficult to justify using on-policy RL methods.  
These limitations have spurred recent research on the LQR problem. 


\smallskip

\wham{Policy gradient algorithms:} 
An inspiration for our work comes from the pioneering contributions
of~\cite{mihailo-2021-tac} and~\cite{fazel_global_2018} who consider the infinite-horizon LQR objective (\eqref{eq:LQ} with $T=\infty$).   With  $x_0$ drawn from a given initial distribution $\mathcal{D}$,   and control policies restricted to the linear form $u_t = Kx_t$,  the optimal control problem reduces to the finite-dimensional static optimization problem: 
\begin{align}
K^\star =\argmin_K
	J(K) = \E \left(\int_{0}^{\infty} x_t\tp Qx_t + u_t\tp Ru_t \,
  dt \right)
	\label{eq:mihailo-LQ-cost} 
\end{align}
where the expectation is over the initial condition.    The authors apply a pure-actor method using ``zeroth order'' methods to approximate gradient descent, much like the early REINFORCE algorithm for RL \cite{sutton-barto}.

In a technical tour de force, a Lyapunov function is obtained to
carry out convergence analysis of the approximate gradient descent algorithm. The
result is surprising because the problem is non-convex in $K$. Error
bounds are obtained to quantify the effect of finite $T$ and the finite
number of iterations of the gradient descent algorithm.  The number of particles
$N_g$ is of the order of the dimension of the system~\cite[Section
VIII]{mohammadi_linear_2021}. 



The trade-off between our algorithm and this prior work is as follows:
While policy optimization methods require multiple iterations with a
small number $N_g$ of particles, the EnKF requires {\em only} a single iteration with
relatively larger number $N$ of particles.  Using the EnKF, it is 
not necessary to have a stabilizing initial gain. 

For a quantitative comparison,  consider using the EnKF algorithm to approximate the infinite-horizon optimal gain (or equivalently the solution to the algebraic Ricatti equation). Choosing $t=0$ in~\eqref{eq:error_formula}, the error is smaller than $\varepsilon$ if the number of particles $N>O(\frac{1}{\varepsilon^2})$ and the simulation time $T>O(\log(\frac{1}{\varepsilon}))$, while the iteration number is one. This is compared with policy optimization approach in~\cite{fazel_global_2018} where the number of particles and the simulation time scales polynomially with $\varepsilon$, while the number of iterations scale as $O(\log(\frac{1}{\varepsilon}))$. This result is later refined in~\cite{mihailo-2021-tac} where the required number of particles and the simulation time are shown to be $O(1)$ and $O(\log(\frac{1}{\varepsilon}))$ respectively (although this result is valid with probability that approaches zero as the number of iterations grow~\cite[Thm. 3]{mihailo-2021-tac}.). 

\renewcommand{\arraystretch}{1.2}
\begin{table}[h]
	\centering 
	\begin{tabular}{|c|c|c|c|}
		\hline 
		Algorithm & particles/samples & simulation time & iterations \\\hline
		EnKF & $O(\frac{1}{\varepsilon^2})$  &  $O(\log(\frac{1}{\varepsilon}))$ &  $1$ \\ \hline
			\cite{fazel_global_2018} & $\text{poly}\left(\frac{1}{\varepsilon}\right)$ &  $\text{poly}\left(\frac{1}{\varepsilon}\right)$ &  $O(\log(\frac{1}{\varepsilon}))$ \\ \hline
				\cite{mihailo-2021-tac}  & $O(1)$  &  $O(\log(\frac{1}{\varepsilon}))$ &  $O(\log(\frac{1}{\varepsilon}))$   \\\hline
	\end{tabular}
\caption{Computational complexity comparison of the algorithms to achieve $\varepsilon$ error in approximating the infinite-horizon LQR optimal gain. }
\end{table}
\renewcommand{\arraystretch}{1}

The overall comparison between the three algorithms appears in \Sec{sec:numerics}. 


\section{Nonlinear extensions}
\label{sec:nonlinear}


We return to the nonlinear optimal control
problem~\eqref{eq:nonlinear_opt_control_problem} in \Sec{sec:intro}. 
Its solution is obtained using a
standard DP argument. 


\wham{Dynamic programming:} For $t\in(0,T)$, the value function
\begin{equation}
v_t(x) := \min_{\{u_s:t\leq s\leq T\}} \int_t^T\left(\half |c(x_s)|^2  +
    \half u_s^\top \RU u_s \right) \ud s + g(x_T)
\end{equation}
From the DP optimality principle, the value function solves the HJB equation
\begin{align}\label{eq:HJB}
\frac{\partial v_t}{\partial t}  + \frac{1}{2}c^2 - \frac{1}{2}\nabla
                    v_t^\top D\nabla v_t + a^\top
\nabla v_t  
  = 0, \quad
v_T (x) = g(x), \quad x\in\Re^d 
\end{align}
where ${D}(x):= b(x)\RU^{-1} b^\top(x)$, and the optimal control input is of the state feedback form $u_t=
\fee_t(x_t)$ where
\begin{equation}
\fee_t(x) =  -  \RU^{-1}b^\top(x)\nabla v_t(x),\quad x\in\mathbb{R}^d
\end{equation}  
is the optimal control policy.  
For the LQ special case, the value function $v_t(x) = \half
x^\top P_t x $ is quadratic and the HJB equation~\eqref{eq:HJB}  reduces to the DRE~\eqref{eq:Ricatti}  for
the matrix 
$P_t$.  

In the following, a mean-field
process is introduced to solve the HJB equation based on the use of a log transformation.


\wham{Log transformation:} Define a probability density as
\[
p_t(x) := \frac{\exp({-v_t(x)})}{\int \exp({-v_t(x)}) \ud x},\quad x\in\Re^d
\]  
In~\ref{app:density}, it is shown that the density solves a backward nonlinear PDE:
\begin{align}
\frac{\partial p_t}{\partial t} & =  p_t (\hP_t- \hat{h}_t) - \nabla \cdot (p_t (a-\nabla \cdot D))  
  -  \frac{1}{2}\nabla^2 \cdot
  (p_tD) \nonumber 
\\
p_T(x) & = \frac{\exp({-g(x)})}{\int \exp({-g(x)}) \ud x}, \quad x\in\Re^d \label{eq:filter}
\end{align} 
where
\[
\hP_t (x):=\frac{1}{2}|c(x)|^2 + (\nabla \cdot a)(x) - \frac{1}{2}\nabla^2 \cdot  D(x) + 
\frac{1}{2}\trace\left(({D(x)})\nabla^2 \log(p_t(x)) \right) 
\]
and $\hat{h}_t:=\int h_t(x) p_t(x) \ud  x$.  



\smallskip

Our objective is to design simulations to sample from $p_t$.  As in
the LQ case, the construction proceeds
in two steps: (i) definition of an exact mean-field process and (ii) its finite-$N$
approximation.

\medskip 
\noindent \textbf{Mean-field process:}  A 
mean-field process $\Ybarbar = \{\Ybarbar_t\in \Re^d:  0 \leq t\leq
T\}$ is defined as follows:
\begin{align}
\ud \Ybarbar_t & = a(\Ybarbar_t) \ud t +  b(\Ybarbar_t) \ud
                 \backward{\etabar}_t + \nabla \cdot D(\Ybarbar_t) \ud
                 t +\vP_t (\Ybarbar_t) \ud t,\nonumber\\ 
\Ybarbar_T &\stackrel{\text{d}}{=}  p_T \label{eq:mfp}
\end{align} 
where ${\etabar}:=\{ {\etabar}_t\in\Re^m:0\leq t\leq T\}$ is a W.P. with covariance $\RU^{-1}$, and 
$\vP_t (\cdot)$ is a vector-field  that solves the 
first order linear PDE
\begin{equation}\label{eq:Poisson}
-\frac{1}{\bar{p}_t(x)} \nabla \cdot(\bar{p}_t(x) \vP_t(x) ) = (\hP_t(x)-\bar{
  \hP}_t), \quad \forall \; x\in\Re^d
\end{equation}
where $\bar{ \hP}_t:=\int h_t(x) \bar{p}_t(x) \ud x$ and $\bar{p}_t$ is the
density of $\Ybarbar_t$ at time $t$.

\medskip

The following proposition relates the density $\bar{p}_t(x)$ of the
mean-field process and the value function $v_t(x)$ of the optimal
control problem.  
 Its proof appears in   \ref{app:density-mean-field}. 

\medskip

\begin{proposition}\label{prop:nonlinear}
Suppose $p_T = \bar{p}_T$.  Then
\[
p_t(x) = \bar{p}_t(x),\quad \forall \; x\in\Re^d,\;0\leq t<T
\]
Consequently, the optimal control law is given by
\[
\fee_t(x) =  \RU^{-1}b^\top (x) \nabla \log \bar{p}_t(x)
\]
\end{proposition}

\wham{Consistency with the LQ setting:} 
With $a(x)=Ax$, $c(x) = C x$, $b(x) = B$, and $p_t={\cal N}(0,P_t^{-1})$ . Then $ \nabla \cdot D= 0 $ and the function $h_t(x)$ simplifies
considerably because  
\begin{align*}
\nabla \cdot a(x) & =(\text{constant}),\quad
\nabla^2 \cdot D = 0,\\ 
\trace\left({D}\nabla^2 \log(p_t(x))
\right) & = (\text{constant})
\end{align*}
Therefore, the right-hand side of the PDE~\eqref{eq:Poisson} is given by
\[
h_t(x)-\bar{h}_t  = \frac{1}{2}|Cx|^2 -  \frac{1}{2}\trace(C^\top C (\bar{S}_t+\bar{n}_t \bar{n}_t^\top ))
\]
It is straightforward to verify that 
\[
\vP_t(x)=\frac{1}{2}S_t C^\top C (x+\bar{n}_t)
\]
solves the PDE~\eqref{eq:Poisson},  from which it follows that  the equation for $\Ybarbar$
reduces to the form described in~\eqref{eq:Ybar}.

The first order PDE~\eqref{eq:Poisson} is well known to arise in the
nonlinear data assimilation literature~\cite{pathiraja2020mckean,daum2017generalized,crisan10}.  One of the issues
with the PDE is that its solution is not unique.  For this reason,
it is useful to consider the gradient form solution such that
$\vP_t(x) =\nabla\phi_t(x)$.  The resulting PDE 
\[
-\frac{1}{\bar{p}_t(x)} \nabla \cdot(\bar{p}_t(x) \nabla\phi_t(x) ) = \hP_t(x)-\bar{
  \hP}_t
\]
is referred to as the Poisson equation,
 where the operator on the
left-hand side is the weighted Laplacian.  Based on assuming a
suitable Poincare inequality, there is a well developed
theory for existence and uniqueness of the solution of the Poisson
equation~\cite[Theorem 1]{laugesen15}. 
 Given its importance in nonlinear filtering, numerical
algorithms for solving the PDE is an area of ongoing
research~\cite{pathiraja2020mckean,taghvaei2019diffusion}.  Approximate formulae for the
solution are also available, e.g., the constant gain approximation
formula~\cite[Example 2]{yang2016}.


\subsection{Dual EnKF for nonlinear systems}
\label{rem:NL-EnKF}

Although one may numerically approximate the solution of the Poisson equation, one
difficulty is that such approximations will require explicit forms of
the vector-fields $a(x)$ and $b(x)$, and will violate
Assumption~\ref{ass:Ass1}.  It is noted that the terms
simplify in the following case:
\begin{enumerate}
\item If $a(x)$ is conservative then $\nabla\cdot a(x) = 0$.  
\item If $b(x)=B$ then $\nabla^2 \cdot D (x) =0$ and $\nabla \cdot D = 0$.
\end{enumerate}    
Upon these simplifications, the mean-field process becomes
\[
\ud \Ybarbar_t = a(\Ybarbar_t) \ud t +  b(\Ybarbar_t) \ud
                 \backward{\etabar}_t +\vP_t (\Ybarbar_t) \ud t
\]
where $\vP_t$ is obtained from solving the PDE~\eqref{eq:Poisson} with 
\[
\hP_t =\frac{1}{2}|c|^2 + 
\frac{1}{2}\trace\left(B^\transpose R^{-1} B\nabla^2 \log(p_t) \right) 
\]
Now, it is natural to consider a Gaussian approximation of the density
$p_t$ whereupon $\hP_t(x) = \frac{1}{2}|c(x)|^2 +
\text{(constant)}$.  This is useful to obtain a dual EnKF
algorithm:
\begin{align*}
 \ud { Y}^i_t & =  \underbrace{a ({ Y}^i_t) 
                    \ud t +  b ({ Y}^i_t) \ud
                    \backward{\eta}^i_t}_{i-\text{th copy of
                model}~\eqref{eq:nonlinear:model-dynamics}}   + {\sf K}_t^{(N)} (\frac{c({ Y}^i_t)
                    + \hat{c}_t^{(N)}}{2}) \ud
   t, \\
{Y}^i_T & \stackrel{\text{i.i.d}}{\sim} \exp({-g_T}),\quad 1\leq i\leq N \nonumber
\end{align*}
where (as before) $\eta^i:=\{\eta_t^i \in \Re^m : i:0\leq t\leq T\}$ is an
 independent copy of $\etabar$, $\hat{c}^{(N)}_{t} := N^{-1}
 \sum_{i=1}^N c(\Ybarbar^i_{t})$, and the gain is a constant matrix:
 \[
 {\sf K}_t^{(N)} = \frac{1}{N-1}\sum_{i=1}^N  ({ Y}^i_t-n^{(N)}_t )(c({ Y}^i_t) - \hat{c}_t^{(N)}  )^\top 
 \]
One may interpret the above as the dual counterpart of the FPF
algorithm with a constant gain approximation~\cite[Example 2]{yang-2016}. 

The optimal control is approximated as in the foregoing via the Hamiltonian,
\begin{align*}
& H^{(N)}(x,\alpha,t) :=  \half
|c(x)|^2 + \half \alpha^\top R \alpha +
 \frac{1}{N} \sum_{i=1}^N  (x^\top  X^i_{t} )(X^i_{t})^\top
  \underbrace{(a(x) + b(x)\alpha)}_{i-\text{th copy of
                model}~\eqref{eq:nonlinear:model-dynamics}}  
\end{align*}
where as before $X_t^i \coloneqq (S^{(N)}_t)^{-1}(Y_t^i - n^{(N)}_t)$.  
Pseudo-code for the
dual EnKF appears in~\ref{app:alg-NL-EnKF}.  

\subsection{Comparison with literature}
\label{sec:NL-comparison}

In the introduction of~\cite{hartmann2012efficient}, the authors write
``{\em Transformations based on an exponential change of measures have
  a rich tradition $\hdots$ and are regularly
  re-discovered}''. Indeed, the pathwise (robust) representation of the nonlinear
filter is based on a log transformation and its link to the HJB equation
is at least as old as the works of~\cite{fleming-1982,fleming1977exit}. 
In the early 2000s, these classical ideas were re-purposed and extended
for the purposes of algorithm design.  There were two sets of ground-breaking contributions: 


\wham{1. Inference as control.} In~\cite{mitter2003}, Mitter and
Newton proposed a dual optimal control formulation of the
nonlinear smoothing equations (see \cite{kim2020optimal} for a recent
review including a discussion of log transformation).    

\wham{2. Control as inference.}
In~\cite{kappen-2005,kappen-2005-jsm}, Kappen described the so called
path integral formulation of optimal control, where the log transformation
is used to convert the HJB equation into a linear equation.  In a
closely related but independent work, Todorov used duality to express
a class of optimal control problems as graphical inference problems~\cite{todorov-2007}.
Both these works continue to impact RL for robotics (a recent
  review is in~\cite{levine-2018}).  

\smallskip

A key idea in these works is the classical connection between  
Kullback-Leibler (KL) divergence and
Bayes' formula: 
Let $\Prob$ denote the law for a stochastic
process $X$ and $\Qrob^z$ denote the conditional law for $X$ given an
observation path $z$ (this is given for inference problems).  Let us construct a controlled process $X^u$
and denote its law as $\Prob^u$ (this is given for control problems).  Assuming $\Prob^u$ is absolutely continuous
with respect to $\Prob$ (denoted $\Prob^u \ll \Prob$), let us
define the objective function as the KL divergence between ${\Prob^u}$
and $\Qrob^z$ as follows:
\begin{equation*}\label{eq:rel-entropy-cost}
\min_{\Prob^u} \quad \E_{{\Prob}^u}\Big(\log \frac{\ud {\Prob}^u}{\ud \Prob}\Big) - \E_{\Prob^u}\Big(\log\frac{\ud \Qrob^z}{\ud \Prob}\Big)
\end{equation*}
In going from inference to control, a model for the controlled process
${X}^u$ is prescribed.  In going from control to inference, the
integral state cost is interpreted as the conditional law $\Qrob^z$
(the second expectation).  Of course,
this places restriction on both the structure of the control system
and the structure of the running cost.  In both Mitter-Newton and in Kappen,
the model structure is as follows:
\[
\ud {X}^u_t = b({X}^u_t,t) \ud t + \sigma({X}^u_t) ( U_t
\ud t + \ud {B}_t)
\]
where $b(\cdot),\sigma(\cdot)$ are $C^1$ vector fields and $\tilde{B}$ is a W.P.
For such a model, $\Prob^u\ll \Prob$ and divergence (the first expectation) equals the quadratic control cost based on the use of the Girsanov
transformation~\cite[Eq.~(35)]{rogers-williams-2000}.  Extension of
these concepts to discrete Markov decision processes (MDP) can be found
in~\cite[Chapter 3]{vhandel-thesis} and is referred to as linearly
solvable MDPs in~\cite{todorov-2007}.  


In~\cite{Kappen2016}, Kappen and Ruiz write ``{\em Despite these elegant theoretical
  results, this idea has not been used much in practice. The essential
  problem is the representation of the controller as a parametrized
  model and how to adapt the parameters such as to optimize the
  importance sampler}''. Indeed, the design of algorithms based on these
ideas remains an important area of research.  

Since our focus is on
inference algorithms for solving optimal control problems, we mention
some salient points:  
The most direct approach is based on exact or approximate inference to
compute the posterior.  Computationally efficient message passing
algorithms for the same are attractive in the linear Gaussian
settings or if the state and action space
is finite~\cite{toussaint-2009,hoffman-2017}.  The optimal control
formulation of the smoothing equations in the linear Gaussian case is completely
classical~\cite[Chapter 15]{kailath2000linear}, as are the message
passing algorithms for the these cases.  In a discrete MDP setting, a relevant 
  example is the posterior policy iteration
  algorithm~\cite[Section ~II-C]{rawlik2013stochastic}.  

For nonlinear SDEs, the
link is again classical -- based on log transformation relating the pathwise
filter and the HJB equation~\cite[Section 3.5]{kim2020optimal}.  The
optimal policy is expressed as a certain Feyman-Kac type expectation
which is approximated using importance sampling.  For MDPs as well, the use of
importance sampling for policy evaluation while sampling from another
(simpler) policy is a standard approach in RL~\cite[Chapter 5.5, 5.7]{sutton-barto}.  It allows the user to explore the state space using an exploratory policy while updating the optimal policy. 



In practice, approximations are necessary.  Based on the KL divergence, a natural
approximation is to parametrize the control policy as $\theta$
  and denote the law of the controlled process as ${P}^\theta$.
  Then policy improvement is obtained using 
\[
\theta \longleftarrow \argmin_{\theta} \; \E_{{P}^\theta}\Big(\log \frac{\ud {P}^\theta}{\ud P}\Big) - \E_{\tilde{P}^\theta}\Big(\log\frac{\ud Q^z}{\ud P}\Big)
\]
The resulting algorithm is referred to as the cross-entropy method in~\cite{Kappen2016} where formulae for the gradient are also obtained and approximated
using importance sampling.  Related concepts and
algorithms appear in a somewhat more general form in~\cite{rawlik2013stochastic} for discrete state-space MDPs.




Given this history, we make the following points to distinguish our
work from this earlier literature:

\wham{1. Log transformation.} 
While our use of the log transformation is same as the path integral approach of
Kappen~\cite{kappen-2005,kappen-2005-jsm},
an important difference is that
for us $p_t$ is a (normalized) probability density.  
The governing equation~\eqref{eq:filter} is nonlinear because of the terms
involving $\bar{h}_t$.  In contrast, the path integral method works
with the un-normalized density whose equation is linear.  The
linearity is crucial for the Feyman Kac formula and its empirical
approximation using importance sampling.  For us, the equation for the
normalized density is necessary because our aim is to construct a
McKean-Vlasov SDE.


\wham{2. Algorithm.} 
The controlled interacting particle
system via a finite-$N$ approximation of the McKean-Vlasov SDE is
original.  It is conceptually and structurally distinct from earlier work, same as the distinction between
important sampling and control-type algorithms in the filtering
context; 
the latter class of algorithms is of much recent origin \cite{reich11,taoyang_TAC12}.  
In
particular, we are not aware of any work using EnKF (or similar
constructions) to
solve an optimal control problem. 

\section{Numerics} \label{sec:numerics}

The performance of the dual EnKF algorithm is numerically evaluated
for three benchmark examples.  In each of the three examples, the
optimal control problem is formulated as an infinite-horizon LQR
problem.  This allows also for a comparison with the state-of-the-art
methods that have focussed on this problem.


In a numerical implementation, the terminal time $T$ is fixed and EnKF is
simulated to obtain an empirical approximation $\{P_t^{(N)}\in\Re^{d\times
  d}:0\leq t<T\}$, typically using $P_T=I$, the identity matrix.
For the sake of comparison, the exact
$\{P_t\in\Re^{d\times d}:0\leq t\leq T\}$ is obtained by numerically 
integrating the DRE~\eqref{eq:Ricatti}.  The stationary solution
$P^\infty$ is obtained as a solution of the ARE using {\tt scipy} package in
Python.  
Pseudo-code is contained in Algorithm \ref{alg:EnKF} of
 \ref{app:alg-NL-EnKF}.  
 				All the code is available on Github \cite{github-repo}.  


\subsection{Linear system with randomly chosen entries}
\label{sec:ex1}
A d-dimensional system is in its controllable canonical form 
\begin{equation*}
A = \begin{bmatrix}
0&1&0&0&\ldots&0\\
0&0&1&0&\ldots&0\\
\vdots & & & & & \vdots\\
a_1 & a_2 & a_3 & a_4 &\ldots & a_d
\end{bmatrix},\quad 
B = \begin{bmatrix}
0\\
0\\
\vdots \\
1 
\end{bmatrix}
\end{equation*}
where the entries $(a_1,\ldots,a_d)\in \mathbb{R}^d$ are i.i.d.
samples from $\normal(0 ,1)$.  The matrices $C,R,P_T$ are identity
matrices of appropriate dimension.  We fix $T=10$, chose the
time-discretization step as $0.02$, and use $N=1000$ particles. 

Figure~\ref{fig:enkf-learns-4} depicts the convergence of the four
entries of the matrix $P_t^{(N)}$  for the case where the state
dimension $d=2$.  
Figure~\ref{fig:enkf-learns-100} depicts the analogous results for
$d=10$. Figures \ref{fig:poles-2} and \ref{fig:poles-10} depict the
open-loop poles (eigenvalues of the matrix $A$) and the closed-loop
poles (eigenvalues of the matrix $(A + BK_0^{(N)})$), for $d=2$ and $d=10$,
respectively. Note that the closed-loop poles are stable, whereas some
open-loop poles have positive real parts.  



	\begin{figure}[h]
		\centering
	\begin{subfigure}{0.49\columnwidth}
         \centering
         \includegraphics[width=1\columnwidth]{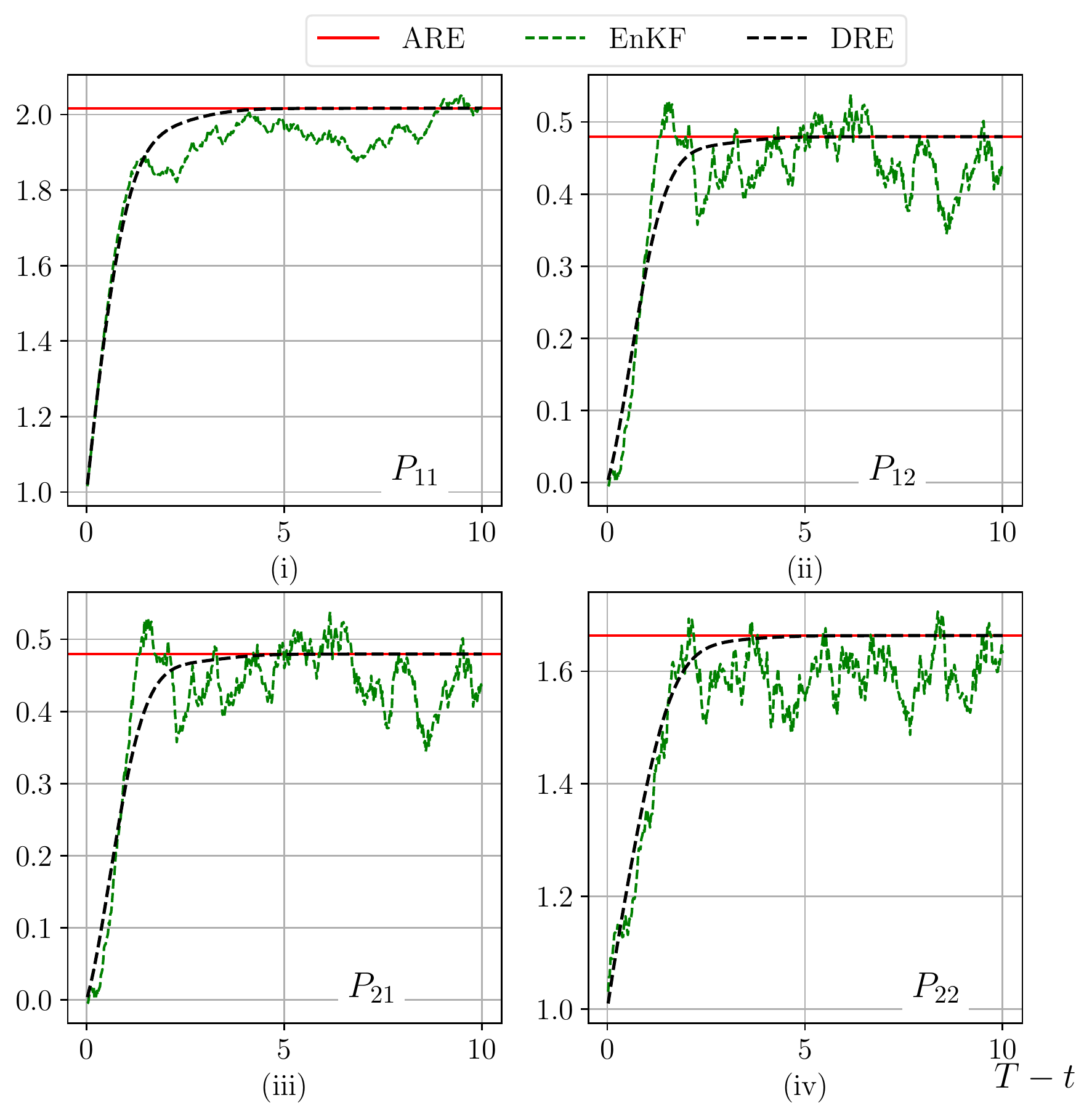}  
         \caption{$d=2$}
         \label{fig:enkf-learns-4}
     \end{subfigure}
     \hfill
	\begin{subfigure}{0.49\columnwidth}
         \centering
         \includegraphics[width=1\columnwidth]{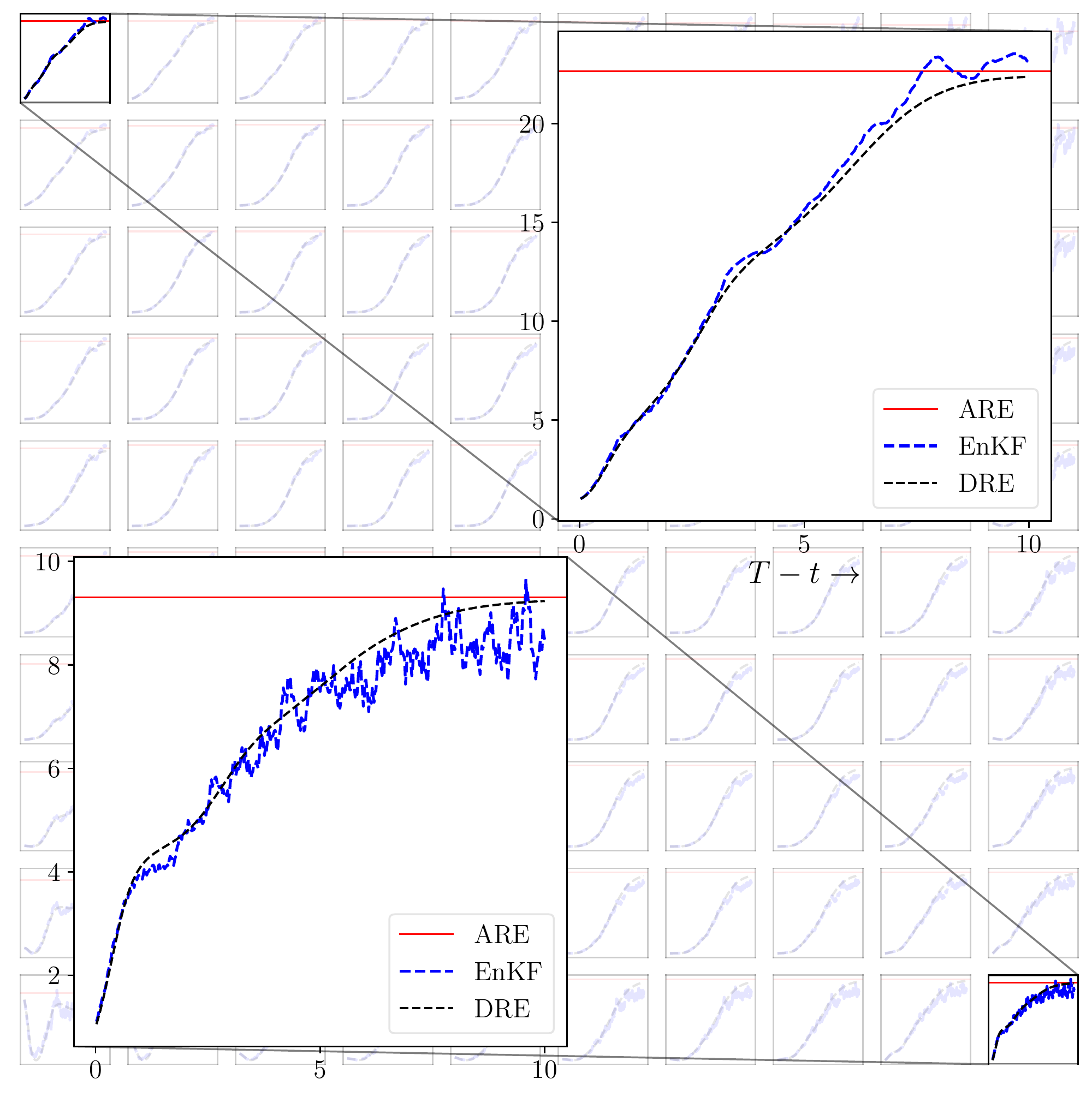}  
         \caption{$d=10$}
         \label{fig:enkf-learns-100}
     \end{subfigure}
%
		\caption{Comparison of the numerical solution obtained
                from the EnKF, the DRE, and the ARE. Note the $x$-axis
              for these plots is $T-t$ for $0\leq t\leq T$. }
		\label{fig:convergence}
		\vspace{-0.1in}
	\end{figure}

\begin{figure}[h]
		\centering
	\begin{subfigure}{0.48\columnwidth}
         \centering
         \includegraphics[width=1\columnwidth]{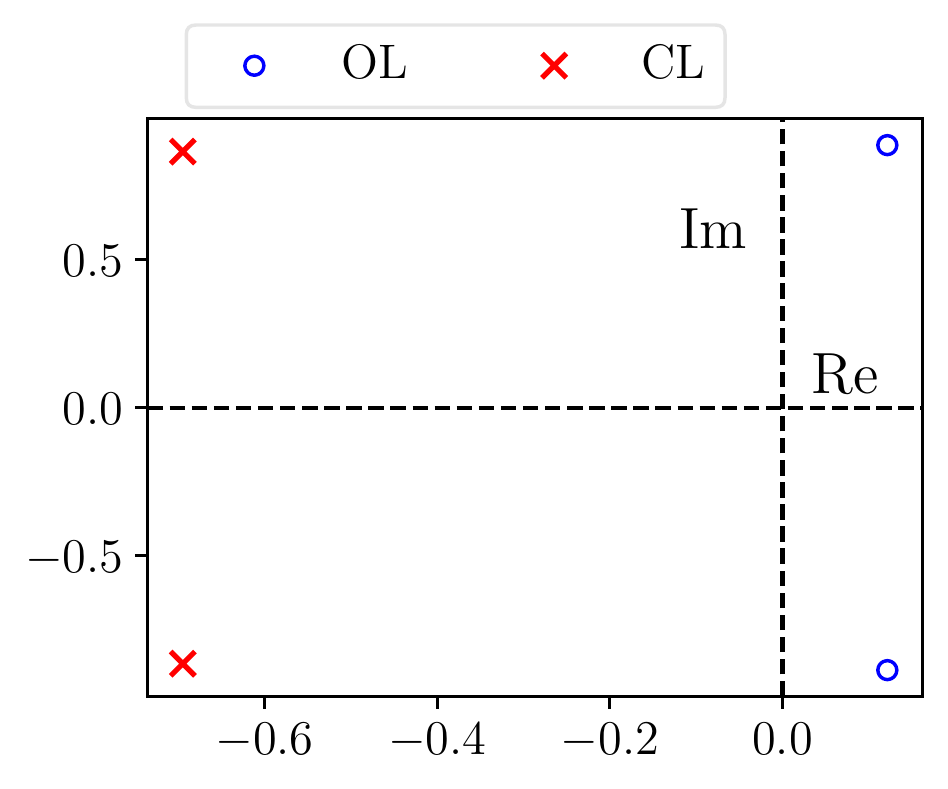}  
         \caption{$d=2$}
         \label{fig:poles-2}
     \end{subfigure}
     \hfill
	\begin{subfigure}{0.48\columnwidth}
         \centering
         \includegraphics[width=1\columnwidth]{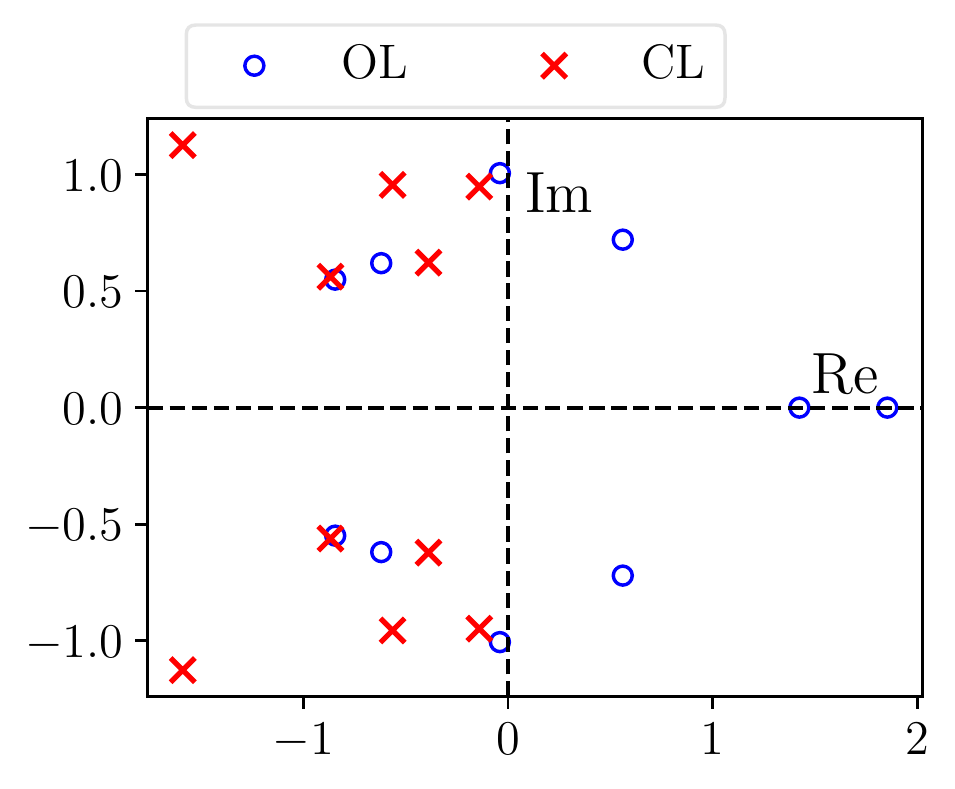}  
         \caption{$d=10$}
         \label{fig:poles-10}
     \end{subfigure}
%
		\caption{Open and closed-loop poles.}
		\label{fig:poles}
		\vspace{-0.1in}
	\end{figure}





	\begin{figure}[h]
		\centering

\includegraphics[width=0.99\columnwidth]{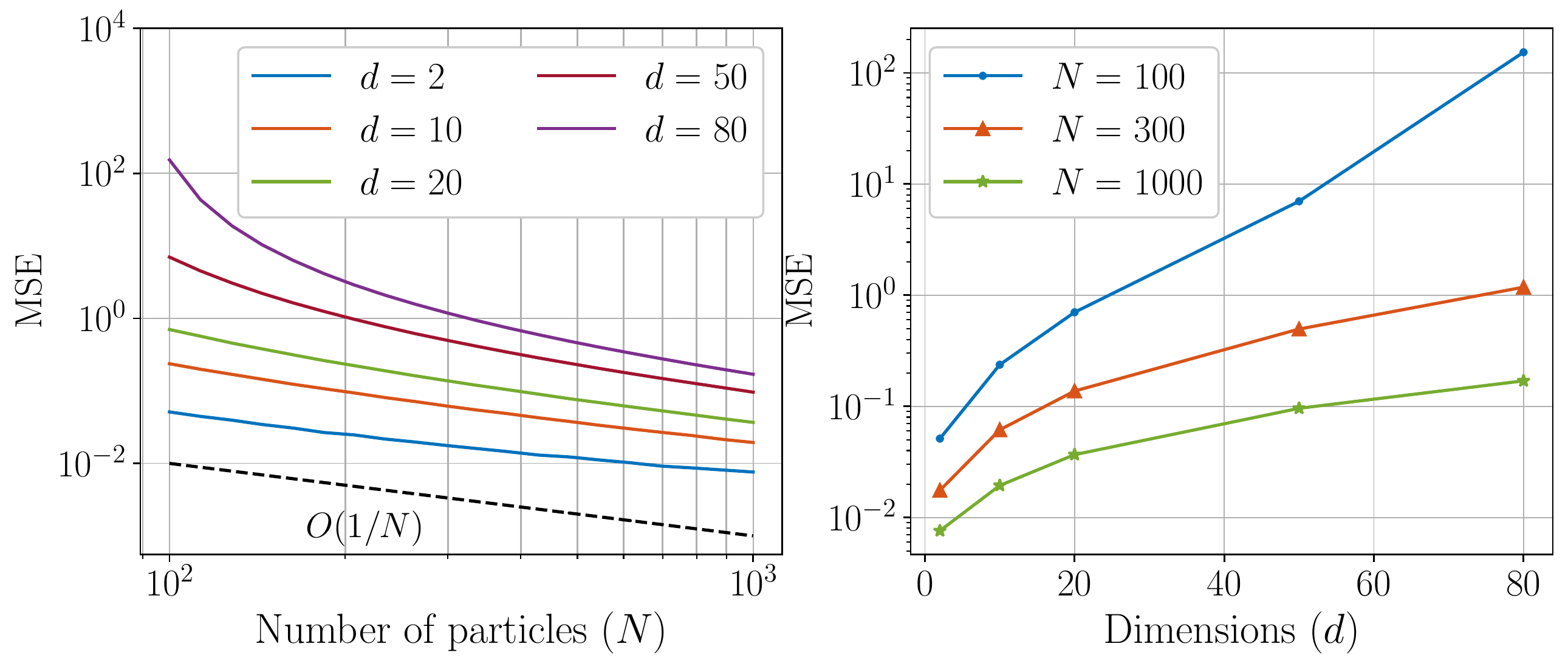}  
 
		\caption{Mean-squared error (MSE) as a function of the
                  number of particles $N$ and system dimension $d$}
        \label{fig:mse}
 	\end{figure}

\begin{figure}[h]
		\centering

\includegraphics[width=0.99\columnwidth]{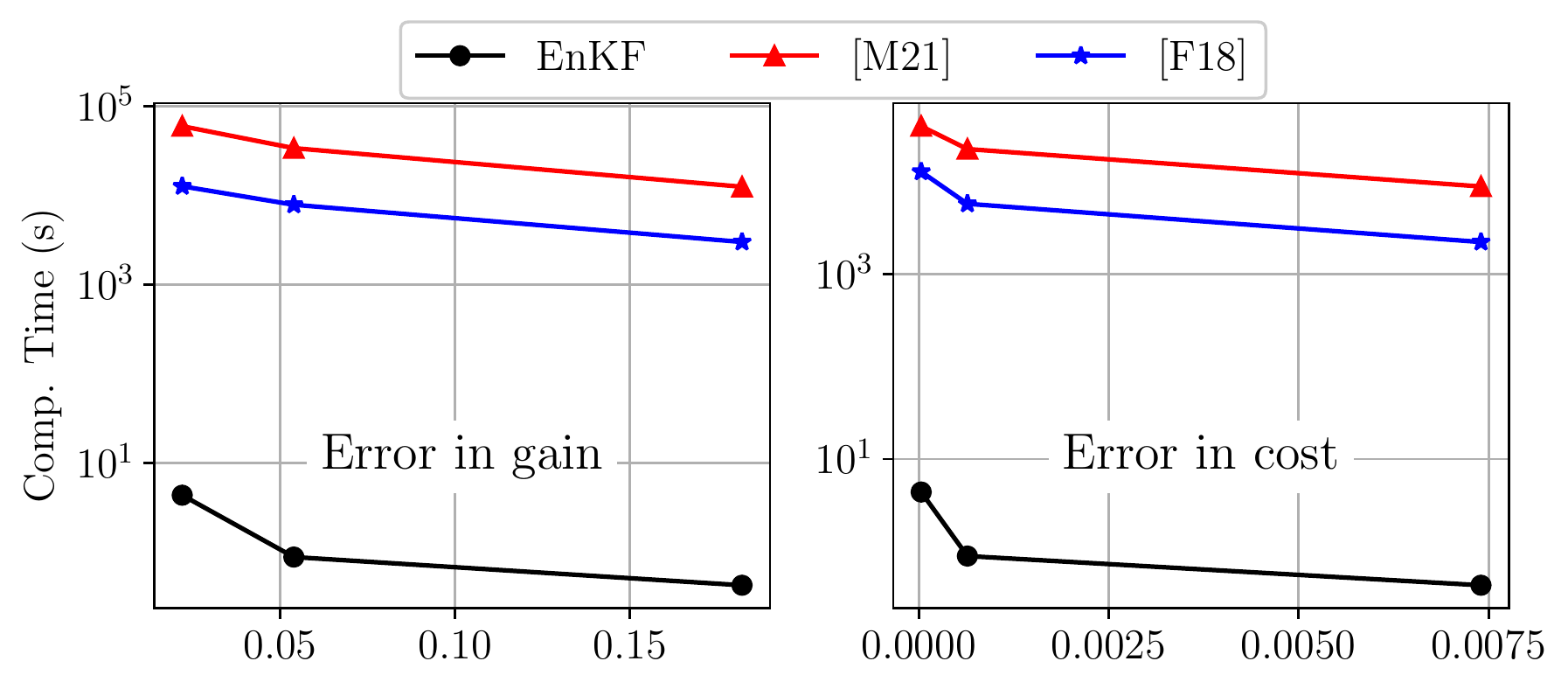}  
 
		\caption{Comparison with algorithms
                  in~\cite{fazel_global_2018} (labeled [F18])
                  and~\cite{mihailo-2021-tac} (labeled [M21]).  The
                  comparisons depict the computation time (in Python) as a
                  function of the relative error in approximating the LQR gain and cost.}\label{fig:comp-to-lit}
 	\end{figure}

\subsection{Mass spring damper system}\label{sec:ex2}

\newcommand{\dsp}{d_s}
\newcommand{\toep}{\mathbb{T}}

We present numerical comparison of EnKF with policy
gradient algorithms in~\cite{mihailo-2021-tac} (denoted as
[M21]) and~\cite{fazel_global_2018} (denoted as [F18]). Comparison is
made on the benchmark spring mass damper example ~\cite[Section
VI]{mohammadi_global_2019}.  Additional details on modeling along with
the numerical
values of various simulation parameters can be found in \ref{app:smd-det}. 


Figure~\ref{fig:mse} depicts the variation of the relative mean-squared error, defined as
\[
\text{MSE} := \frac{1}{T} \E \left( \int_{0}^{T} \frac{\|
  P_t - P_t^{(N)} \|_F^2}{\| P_t \|_F^2} \: \ud t \right)
\]  
The figure depicts two trends: the
  $O(\frac{1}{N})$ decay of the MSE as $N$ increases (for $d$ fixed), which is an illustration of the error bound \eqref{eq:error_formula},
  and a plot of the MSE as a function of dimension $d$ (for $N$ fixed).

A side-by-side comparison with [F18] and [M21] is depicted in
Fig.~\ref{fig:comp-to-lit}.  The comparison is for the following
metrics (taken from \cite{mihailo-2021-tac}):
\[
\text{error}^{\text{gain}}= \frac{\| K^{\text{est}} - K^{\infty} \|_F}{\|K^{\infty}\|_F} \, ,
\quad \text{error}^{\text{value}} =\frac{ c^{\text{est}} - c^{\infty} }{c^{(N)}_{\text{init}} - c^{\infty} }
\]
where the LQR optimal gain $K^{\infty}$ and the optimal value
$c^{\infty}$ are computed from solving the ARE.  The value
$c^{(N)}_{\text{init}}$ is approximated using the initial gain $K=0$
(Note such a gain is not necessary for EnKF).  Because [F18]
is for discrete-time system, we use the Euler approximation to obtain
a discrete-time model.  Such an approximation is consistent with our
choice of numerical integration in Algorithm \ref{alg:EnKF}.  

To obtain the relationship between the error and computational time,
the number of particles $N$ is varied in the EnKF algorithm while the
number of gradient descent steps is changed in [M21] and [F18].

In the numerical experiments, the dual EnKF is found to be
significantly more
computationally efficient--by two orders of magnitude or more. 
Comparison was carried out for a range of $d$ and is qualitatively similar, see
 \ref{app:comparison}.
The main reason for the order of
magnitude improvement in computational time is as follows:  An EnKF
requires only a single iteration over a fixed time-horizon
$[0,T]$. 
We found that the number of particles 
($N$) for the EnKF algorithm is typically one or two orders of magnitude larger than $N_g$. Since our algorithm is designed to be written as a matrix vector multiplication, vectorization features of the  {\tt numpy} package
in Python yield significant gains in computational time.  In contrast, [F18]
and [M21] require several steps of gradient descent, with each step requiring an evaluation of the LQR cost, and because these operations
must be done serially, these computations are slower.
In our comparisons, the same time-horizon $[0,T]$ and discretization
time-step $\Delta t$  was used for all the algorithms.  It is certainly
possible that some of these parameters can be optimized to improve the
performance of the other algorithms. In particular, one may consider
shorter or longer time-horizon $T$ or use parallelization (over the
$N_g$ copies) to speed up the gradient calculation.  Codes are made
available on Github for interested parties to independently verify
these comparisons~\cite{github-repo}.

\subsection{Nonlinear cart-pole system}\label{sec:ex3}

Figure~\ref{fig:ivp} depicts the closed-loop trajectories of a
four-dimensional nonlinear cart pole model.  The control acts as external
force applied to the cart.  The four-dimensional state for the system
is $(\theta,x, \dot{\theta},\dot{x})$, where $\theta\in {\sf S}^1$ (the circle) is the
angle of the pole (pendulum) as measured from the stable equilibrium,
$x\in\Re$ is the displacement of cart along the horizontal. The control objective is
to balance the pole -- stabilize the system
at the inverted equilibrium $(\pi,0,0,0)$, assuming full state
feedback.  (See \ref{app:cp-det} for details on the model
parameters and their numerical values).

For the purposes of control design, the nonlinear system is
first linearized at the desired equilibrium and an LQR problem is formulated based on~\cite[Chapter 3.2]{tedrake-notes}.
The (nonlinear) dual EnKF is used to approximate the
optimal control law which is numerically evaluated on the fully nonlinear model.  Figure~\ref{fig:ivp} depicts the numerically obtained results.  It was
found that reasonable levels of performance is obtained with as few as $N=10$
particles.  With $N=1000$ particles, the closed-loop trajectories are virtually
indistinguishable from the exact optimal control solution.


	\begin{figure}[h]
		\centering
\vspace*{-0.1in}
	\includegraphics[width=0.99\columnwidth]{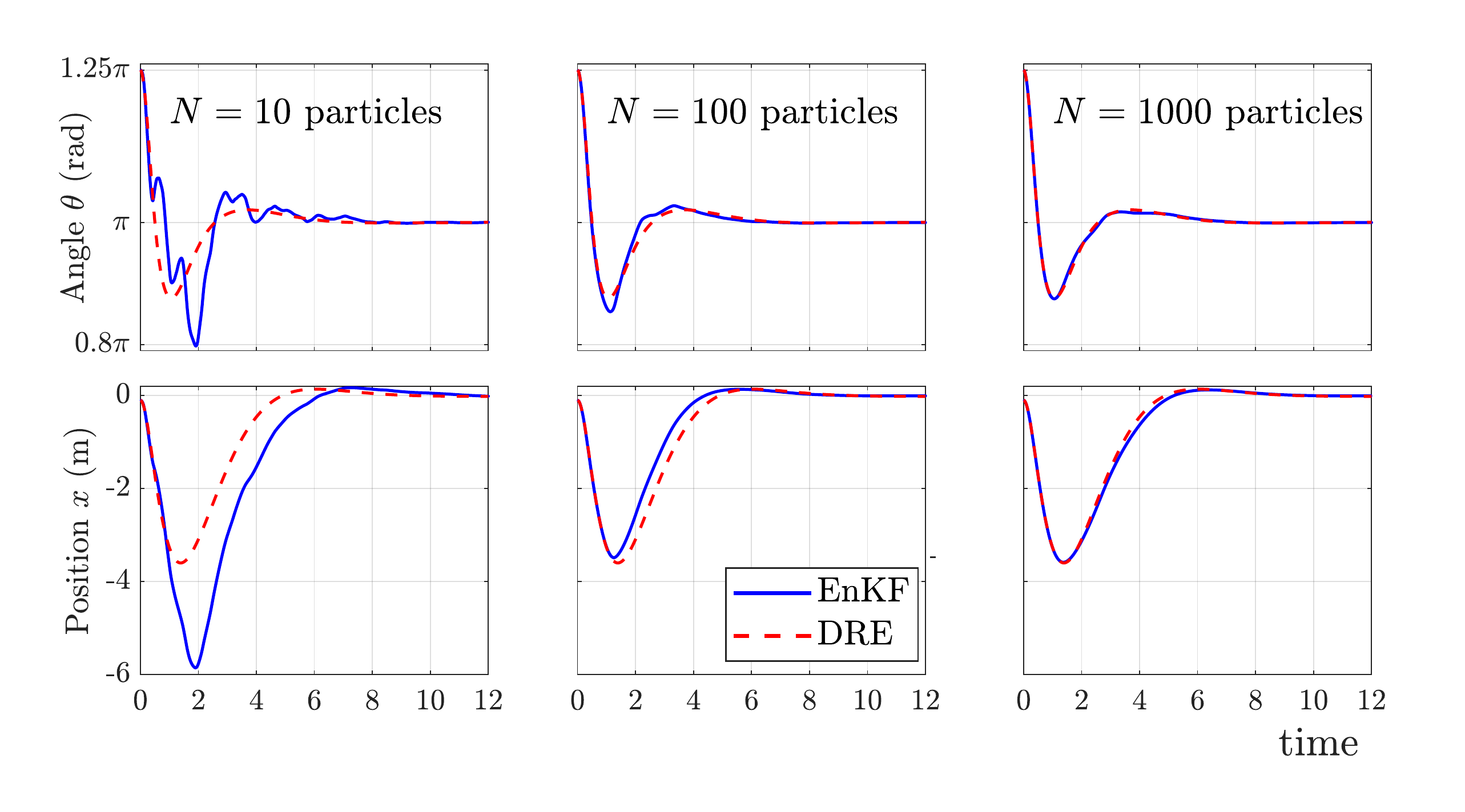}  
\vspace{-0.3in}
		\caption{Trajectories of the closed-loop nonlinear cart pole.}\label{fig:ivp}
\vspace{-0.2in}
	\end{figure}

\section{Conclusions}

In this paper, we present a new class of algorithms for learning
optimal policies using simulations.  A key message is
that log transforms combined with mean field techniques can lead to
simulation based methods for optimal policy approximation.  We have
demonstrated this for LQ in full detail, and shown how the
techniques generalize to nonlinear systems.       

There are two key innovations: (i) the representation of the unknown
value function in terms of the statistics (variance) of a suitably
designed process; and (ii) design of interactions between simulations
for the purposes of policy optimization.  

We fully believe that the two key innovations may be useful for many
other types of models including MDPs and partially observed problems.
For policy evaluation, use of Monte Carlo techniques is already
standard.  It is shown in this paper is that by designing careful interaction
amongst simulations, one can also solve the policy optimization
problem.  

Another notable aspect is the learning rate.  
Because the $N=\infty$ limit is
exact tor the LQR problem, the proposed algorithms yields a learning
rate that closely approximates the exponential rate of convergence of
the solution of the DRE.  This is rigorously established with the aid
of error bound~\eqref{eq:error_formula} (although such an analysis is conservative).  In numerical examples, this property is 
shown to lead to an 
order of magnitude better performance than the state-of-the-art
algorithms.  

Given the enormous success of EnKF in data assimilation~\cite{reich2015probabilistic,evensen2006}, the
contributions of this paper potentially open up new opportunities for
RL.
It is our hope that the paper will engender new synergies between
the data
assimilation and the RL communities.

\appendix

\section{Proof of \Prop{prop:Y-exactness}}
\label{app:prop-y-exact}
The equation for the mean $\bar{n}_t$ is obtained by taking the expectation of SDE~\eqref{eq:Ybar},
\begin{equation*}
\ud \bar{n}_t = (A +  \bar{S}_tC^\top C) \bar{n}_t\ud t
\end{equation*}
Because $\bar{n}_T = 0$, we have $\bar{n}_t=0$ for all $ t \in [0,T]$. 

The equation for the covariance $\bar{S}_t$ is obtained by writing the SDE for the error $e_t := \bar{\mathcal Y}_t  - \bar{n}_t$:
\begin{align*}
\ud e_t &= (A+\frac{1}{2}\bar{S}_t C^\top C) e_t  \ud t + B \ud \backward{\etabar}_t,
\end{align*}  
Using the It\^o rule for $e_te_t^\top$,
\begin{align*}
&\ud (e_te_t^\top) =  B \ud \backward{\etabar}_te_t^\top + e_t (B\ud \backward{\etabar}_t)^\top - B R^{-1}B^\top \\ & \quad + (A+\frac{1}{2}\bar{S}_t C^\top C) (e_te_t^\top)   \ud t +  (e_te_t^\top)(A+\frac{1}{2}\bar{S}_t C^\top C) ^\top   
\end{align*}  
The It\^o correction term appears with a negative sign because the SDE involves a backward Wiener process ~$\backward{\etabar}_t$~\cite[Sec. 4.2]{nualart1988stochastic}. Taking an expectation yields the following equation for $\bar{S}_t$:
\begin{equation*}
\frac{\ud}{\ud t} \bar{S}_t= (A+\frac{1}{2}\bar{S}_t C^\top C)\bar{S}_t+  \bar{S}_t(A+\frac{1}{2}\bar{S}_t C^\top C) ^\top  - B R^{-1}B^\top
\end{equation*}  
The SDE is identical to the SDE for $S_t$. Because $\bar{S}_T =S_T$,
we  have $\bar{S}_t = S_t$  for all $ t \in [0,T]$. The conclusion
that $\Ybarbar_t$ is Gaussian follows from the fact that with
$\bar{n}_t=n_t$ and $\bar{S}_t=S_t$, the SDE for $\Ybarbar_t$ is an
Ornstein-Uhlenbeck SDE with a Gaussian terminal condition.

The proof for the rest of proposition is straightforward.  By definition,
\begin{align*}
\Expect[\Ybar_t] &= \Expect[\bar{S}_t^{-1}(\Ybarbar_t-\bar{n}_t)]=\bar{S}_t^{-1}(\bar{n}_t-\bar{n}_t)=0\\
\Expect[\Ybar_t\Ybar_t^\top] &= \Expect[\bar{S}_t^{-1}(\Ybarbar_t-\bar{n}_t)(\Ybarbar_t-\bar{n}_t)^\top\bar{S}_t^{-1} ] =  \bar{S}_t^{-1}  = S_t^{-1} = P_t
\end{align*} 

\section{Error analysis}
\label{app:error-proof}

\newcommand{\OmegaN}{\Omega^{(N)}}
\newcommand{\BB}{\Sigma_B}

{\noindent \bf Notation:} Let $S^d_{+} \subset S^d \subset \R^{d \times d}$ denote the set of symmetric positive definite matrices and symmetric matrices respectively. Let $\langle Q_1, Q_2\rangle \coloneqq \trace(Q_1Q_2^\top)$ denote the Frobenius inner product, and $||\cdot||_{F} \coloneqq \sqrt{\langle Q_1, Q_1 \rangle}$ denote the Frobenius inner product for $Q_1,Q_2 \in \R^{d \times d}$.  

The objective is to study the error between the empirical covariance of the particles $\SN_t$ and its mean-field limit $S_t$.  To simplify the presentation, we use the time-reversed quantitative $\Omega^{(N)}_t:=S^{(N)}_{T-t}$ and $\Omega_t:=\bar{S}_{T-t}$. 
According to the Proposition~\ref{prop:Y-exactness}, $\Omega_t$ satisfies the Riccati equation
\begin{equation}\label{eq:Sbar-t}
\frac{\ud}{\ud t} \Omega_t =  \Ricc(\Omega_t):=-A\Omega_t  - \Omega_t A^\transpose  - \Omega_tC^\transpose C \Omega_t + \BB,
\end{equation}  
where $\BB:=BR^{-1}B^\top$. The time-evolution for $\Omega^{(N)}_t$ is obtained by the application of the It\^o rule to its definition~\cite[Prop. 4.2]{bishop2020mathematical}
	\begin{equation}\label{eq:SNt}
	\ud \OmegaN_t = \Ricc( \OmegaN_t)\ud t + \frac{1}{\sqrt{N}}\ud M_t,
	\end{equation} 
	where $\{ M_t : t \ge 0 \}$ is a martingale 
	given by
	\begin{align*}
	dM_t &= \frac{1}{\sqrt{N}}\sum_{i=1}^{N}F^{i}_t(B\ud{\eta}_t^i)^\transpose \!\!+ \!\!B \ud {\eta}_t^i (F^i_t)^\transpose, \quad F_{T-t}^i \coloneqq X_{t}^i\! -\! n_{t}^{(N)}
	\end{align*}
	with quadratic variation
		\begin{align*}
 \ud\langle M\rangle_t = \trace(\BB) \OmegaN_t + \BB \trace( \OmegaN_t) + \BB  \OmegaN_t  +\OmegaN_t  \BB 
	\end{align*}

The error analysis is based on a sensitivity  analysis of the Riccati equation. Let $\phi(t,Q)$ denote the semigroup associated with the Riccati equation such that for any positive definite matrix $Q \in S^d_+$, 
\begin{equation*}
	\frac{\partial \phi}{\partial t} (t,Q) = \Ricc(\phi(t,Q)),\quad \phi(0,Q)=Q. 
\end{equation*}
We define the first-order and the second-order derivatives which are the linear and bilinear operators $\frac{\partial \phi}{\partial Q}(t,Q):S^d \to S^d$ and $\frac{\partial^2 \phi}{\partial Q^2}(t,Q):S^d\times S^d\to S^d$ respectively that satisfy
\begin{align*}
\frac{\partial \phi}{\partial Q}(t,Q) (Q_1)   &= \left.\frac{\ud}{\ud \epsilon}\right\vert_{\epsilon=0} \phi(t,Q+\epsilon Q_1) \\
\frac{\partial^2 \phi}{\partial Q^2}(t,Q)(Q_1, Q_1)  &= \left.\frac{\ud^2}{\ud \epsilon^2} \right\vert_{\epsilon=0}\phi(t,Q+\epsilon Q_1) .
\end{align*}
We also let $\| \frac{\partial \phi}{\partial Q}(t,Q) \|_{F,F}$ and $\| \frac{\partial^2 \phi}{\partial Q^2}(t,Q)\|_{F,F}$ denote the induced-norm of these operators with respect to the Frobenius norm.   The following lemma expresses the error as a stochastic integral that involves the semigroup. 
\begin{lemma}
	Consider $\Omega_t$ and $\OmegaN_t$ defined in~\eqref{eq:Sbar-t} and \eqref{eq:SNt} respectively. Then
	\begin{equation}
	\begin{aligned}
		\OmegaN_t-&\Omega_t= \frac{1}{\sqrt{N}}\int_0^t \frac{\partial \phi}{\partial Q}(t-s,\OmegaN_s) (\ud M_s) \\&+ \frac{1}{2N}\int_0^t \frac{\partial^2 \phi}{\partial Q^2}(t-s,\OmegaN_s)(\ud M_s, \ud M_s) \\&+ \phi(t,\OmegaN_0) - \phi(t,\Omega_0) 
	\end{aligned}
\label{eq:integral-representation}
\end{equation}
\end{lemma}
\begin{proof}
	The proof follows by expressing the difference 
\begin{align*}
	\Omega^{(N)}_t -& \Omega_t =  \phi(0,\Omega^{(N)}_t)  - \phi(t,\Omega_0)\\
	&=   \phi(0,\Omega^{(N)}_t)  - \phi(t,\Omega^{(N)}_0)  +\phi(t,\Omega^{(N)}_0) - \phi(0,\Omega_0)
	\\&= \int_0^t\ud_s\phi(t-s,\Omega^{(N)}_s) +  \phi(t,\Omega^{(N)}_0) - \phi(t,\Omega_0),
\end{align*}
and evaluating the differential 
\begin{align*}
	\ud_s \phi(t-s,\Omega^{(N)}_s) &=   -\frac{\partial \phi}{\partial t} 
	(t-s,\Omega^{(N)}_s) \ud s + \frac{\partial \phi}{\partial Q}(t-s,\Omega^{(N)}_s) (\ud \Omega^{(N)}_s)  \\& \quad +  \frac{1}{2}\frac{\partial^2 \phi}{\partial Q^2}(t-s,\Omega^{(N)}_s) (\ud \Omega^{(N)}_s,\ud \Omega^{(N)}_s), 
\end{align*}
and using the identity $ \frac{\partial \phi}{\partial t}(t,Q) =  \frac{\partial \phi}{\partial Q}(t,Q) (\text{Ricc}(Q))$. 
\end{proof}
\medskip

The preceding lemma can be viewed as the extension of the Alekseev-Gr\"obner formula to matrix-valued stochastic differential equations~\cite{del2019forward}. The explicit form of this expression appears in~\cite[Sec. 5.3]{bishop2019stability}. 

The error bound follows from uniform bounds on the terms involved in the integral~\eqref{eq:integral-representation}. Such uniform bounds are available if the Riccati equation enjoys the following stability property. 
\begin{assumption}
	Consider the  semigroup corresponding to the Riccati equation~\eqref{eq:Sbar-t}. There are positive constants $c_1$, $c_2$, and $\lambda$ such that $\forall Q \in S^d_+$: 
	\begin{align*}
		\|\frac{\partial \phi}{\partial Q}(t,Q)\|_{F,F} \leq c_1e^{-2\lambda t},\quad 	\|\frac{\partial^2 \phi}{\partial Q^2}(t,Q)\|_{F,F} \leq c_2e^{-2\lambda t}. 
	\end{align*}
\label{assumption:Ricatti}
\end{assumption}
These bounds are directly related to the exponential stability of the closed-loop linear system under optimal feedback control~\cite[Sec. 2]{bishop2020mathematical}. The exponential decay holds when the linear system is controllable and observable. However, the fact that the constants $c_1$ and $c_2$ are uniform among all initial matrices $Q$ is still open. See~\cite{bishop2020mathematical,bishop2017stability} for detailed analysis of the Riccati equation where these uniform bounds are shown to hold under the additional assumption that  the matrix $C$ is full-rank.



\begin{proposition}\label{prop:error-analysis-appendix}
	Let $\bar{S}_t$ be the mean-field covariance defined in~\eqref{eq:Sbar} and $S^N_t$ be the empirical covariance of the particles defined in \eqref{eq:SNt}. Then, under Assumption~\ref{assumption:Ricatti}, the error between $S^{(N)}_t$ and $\bar{S}_t$ satisfies the upper-bound 
\begin{equation}\label{eq:error_formula-appendix}
\Expect[\|S^{(N)}_{t}-\bar{S}_t\|_F] \leq
\frac{C_1}{\sqrt{N}} + C_2e^{-2\lambda (T-t)} \Expect[\|\SN_T-\Sbar_T\|_F],
\end{equation}
where $C_1,C_2$ are  time-independent positive constants.
\end{proposition}
\begin{proof}
	Using~\eqref{eq:integral-representation} and the triangle inequality,  the expected norm of the difference satisfies
	\begin{align*}
		\Expect[\|\OmegaN_t-\Omega_t\|_F] \leq \frac{R_1}{\sqrt{N}} + \frac{R_2}{2N} + R_3
	\end{align*}
where 
\begin{align*}
	R_1 &= \Expect \left[\left\|\int_0^t \frac{\partial \phi}{\partial Q}(t-s,\Omega_s)(\ud M_s)\right\|_F\right]\\ 
	R_2 & = \Expect  \left[\int_0^t \left\|\frac{\partial^2 \phi}{\partial Q^2}(t-s,\Omega_s)(\ud M_s,\ud M_s)\right\|_F\right]  \\ 
	R_3 & = \Expect \left[\left\| \phi(t,\OmegaN_0) - \phi(t,\Omega_0)\right \|_F\right]
\end{align*}
The first term
\begin{align*}
	R_1 
	&\leq  \left[\Expect\left[ \left\|\int_0^t \frac{\partial \phi}{\partial Q}(t-s,\Omega_s)(\ud M_s)\right\|_F^2 \right]\right]^\half\\
	&=\left[\int_0^t \Expect \left[\left\| \frac{\partial \phi}{\partial Q}(t-s,\Omega_s)(\ud M_s)\right\|_F^2\right] \right]^\half\\
	& \leq \left[\int_0^t \Expect \left[\| \frac{\partial \phi}{\partial Q}(t-s,\Omega_s)\|_{F,F}^2 \|\ud M_s\|_F^2\right] \right]^\half   \\
	&\leq \left[\int_0^t 4c_1^2e^{-4\lambda (t-s)}\trace(\Sigma_B)\Expect[\trace(\OmegaN_s)]\ud s\right]^\half
\end{align*}
where we used Jensen's inequality in the first step, It\"o isometry in the second step, and Assumption~\ref{assumption:Ricatti} in the last step. The second term,
\begin{align*}
	R_2 &\leq \Expect  \left[\int_0^t \|\frac{\partial^2 \phi}{\partial Q^2}(t-s,\Omega_s) \|_F \| \ud M_s\|^2_F\right]\\
	&\leq  \int_0^t 4c_2e^{-2\lambda (t-s)}\trace(\Sigma_B) \Expect[\trace(\OmegaN_s) ]\ud s
\end{align*}
	where we used Assumption~\ref{assumption:Ricatti}. The third term,
	\begin{align*}
	R_3 \leq c_1e^{-2\lambda t}\Expect[\|\OmegaN_0 - \Omega_0\|_F]
	\end{align*} 
because of the bound on the first derivative in  Assumption~\ref{assumption:Ricatti}. Upon using the bound $\Expect[\trace(\OmegaN_t)] \leq \trace(\Sigma_t) \leq \sup_{t\geq  0} \trace(\Sigma_t) =:\sigma^2$ from~\cite[Thm. 5.2]{bishop2020mathematical}, we conclude
\begin{align*}
		\Expect[\|\OmegaN_t-\Omega_t\|_F] \!\leq \!(c_1 \!+ \!c_2\sqrt{\epsilon})\sqrt{\epsilon} \!+ \!c_1e^{-2\lambda t} \Expect[\|\OmegaN_0 - \Omega_0\|_F]
	\end{align*}
where $\epsilon: = \frac{\sigma^2\trace(\Sigma_B)}{\lambda N}$. Changing $t$ to $T-t$ concludes the proof.  
\end{proof}

%
%

\section{Evolution of density in \eqref{eq:filter}}
\label{app:density}

By definition, the probability density 
\begin{align*}
p_T(x) & = \frac{\exp({-g(x)})}{\int \exp({-g(x)}) \ud x} \, , \quad x\in\Re^d
\end{align*}
Write $v_t  = -\log(p_t) + \beta_t$ where $\beta_t \coloneqq \log (\int p_t(x) dx)$ is a
time-dependent constant to ensure $p_t$ is normalized.  In terms of
$p_t$ and $\beta_t$, the HJB equation~\eqref{eq:HJB} for $v_t$ is written as 
\begin{multline*}
- \frac{1}{p_t}\frac{\partial p_t}{\partial t}  + \dot{\beta}_t + \frac{1}{2} |c|^2 - \frac{1}{p}  a^T \nabla p - \frac{1}{2p_t} \trace(D\nabla^2 p_t ) -\frac{1}{2} \trace((Q-D)\nabla^2 \log(p_t) ) = 0
\end{multline*}
where we used $ \nabla^2 \log(p_t) = \frac{1}{p_t} \nabla^2p_t  - \frac{1}{p_t^2}\nabla p_t \nabla p_t^\top$. 
Multiplying by $p_t$ yields 
\begin{equation*}
\frac{\partial p_t}{\partial t}  = (h_t + \dot{\beta}_t)p_t -\nabla \cdot(p_ta)   + \nabla \cdot(p_t \nabla \cdot D ) - \frac{1}{2}  \nabla^2 \cdot(p_tD)
\end{equation*}
where we used 
\begin{align*}
h_t := \frac{1}{2}|c|^2 + & \nabla \cdot a - \frac{1}{2}\nabla^2 \cdot D + \frac{1}{2}\trace((D-Q)\nabla^2 \log(p_t)) \\
a^T\nabla p _t  &=\nabla \cdot(p_t a)-  p_t\nabla \cdot a\\
\trace(D\nabla^2p_t)& = \nabla^2 \cdot(p_tD) - 2\nabla \cdot(p_t \nabla \cdot D ) + p_t \nabla^2 \cdot D 
\end{align*}
Noting $\int \frac{\partial p_t}{\partial t} \ud x = 0$, we obtain 
\begin{equation*}
\dot{\beta}  = -\int h_t(x) p_t(x)\ud x = -\hat{h}_t
\end{equation*} 
which in turn gives the PDE~\eqref{eq:filter} for $p_t$. 

\section{Proof of \Prop{prop:nonlinear}}
\label{app:density-mean-field}

The proof for $\bar{p}_t = p_t$ follows from showing that the
evolution  equation for $\bar{p}_t$ and $p_t$ are identical. Consider
the SDE~\eqref{eq:mfp}.  The evolution equation for the density
$\bar{p}_t$ is the Fokker-Planck equation:
\begin{equation*}
\frac{\partial \bar{p}_t}{\partial t}  = - \nabla \cdot (\bar{p}_t a)  - \nabla\cdot(\bar{p}_t\nabla \cdot D)  - \nabla \cdot (\bar{p}_t\mathcal V_t)  -\frac{1}{2} \nabla^2\cdot (\bar{p}_tD)
\end{equation*} 
where the diffusion term $\frac{1}{2} \nabla^2\cdot (\bar{p}_tD)$
appears with a negative sign because $\backward{\bar{\eta}}_t$ is a
backward Wiener process.  

It is easily see that if the vector-field $\vP_t (\cdot)$ solves the
PDE~\eqref{eq:Poisson} then the evolution equations for $p_t$ and
$\bar{p}_t$ are identical.  

\section{Algorithm for implementing nonlinear dual EnKF}

\label{app:alg-NL-EnKF}

The algorithm to approximate the optimal control policy for \eqref{eq:nonlinear_opt_control_problem} is divided into an online and offline component. 

\wham{Offline algorithm.} (Algorithm \ref{alg:P}) to compute $\{P_t^{(N)}:0\leq
t\leq T\}$.  It is based on the
finite-$N$ approximation of the dual EnKF~\eqref{eq:EnKF-Yi}.  
For a numerical solution of the SDE, we
use the simplest Euler scheme which can be swapped with a higher order
scheme.   
 
\wham{Online algorithm.} (Algorithm \ref{alg:EnKF}) to compute the optimal control for
a given state $X_t=x$ at time $t$.  In addition to the simulator, this
algorithm also requires $P_t^{(N)}$ computed from the offline
  algorithm.  It is based on minimizing the Hamiltonian function.

The algorithm is described for the general nonlinear case.  The LQ is
the special case when $f(x,u)=Ax+Bu$ and $c(x) = Cx$.   

In a numerical implementation of the offline algorithm, there are two
sources of error: (i) because of finite-$N$ approximation;
and (ii) because of time-discretization step size $\Delta t$.  The
first type of error scales as $O(\frac{1}{\sqrt{N}})$ as shown in the
bound~\eqref{eq:error_formula}.  For SDEs, the second type of error scales as
$O(\Delta t)$ using the Euler
scheme~\cite{platen1999numerical}.  

\begin{algorithm}
	\caption[Offline]{\textbf{[offline]} EnKF algorithm to approximate
      $\{P_t:0\leq t\leq T\}$}
	\label{alg:P}
	\begin{algorithmic}[1]
		\REQUIRE Simulation time $T$, simulation step-size $\stepsize$, number of particles $N$, simulator $f(x,u) =a(x) + b(x)u$, terminal cost $g_T$, cost function $c(x)$, and control cost matrix $R$.
		\RETURN $\{P^{(N)}_k(\cdot) : k=0,1,2,\ldots,\frac{T}{\stepsize} -1\}$
		\STATE  $T_F = \frac{T}{\stepsize}$ 
		\STATE  Initialize $\{{Y}^i_{T_F}\}_{i=1}^N\iid \exp({-g_T})$
		\STATE calculate $n^{(N)}_{T_F} = N^{-1}\sum_{i=1}^N { Y}^i_{T_F}$ 
		\FOR{$k=T_F$  to $1$} 
		\STATE Calculate 
		$\hat{c}^{(N)}_k= N^{-1}\sum_{i=1}^N c({ Y}^i_k)$  
		\STATE Calculate	$ M^{(N)}_k= (N-1)^{-1}\sum_{i=1}^N ({ Y}^i_k-n^{(N)}_k)(c({ Y}^i_k) - \hat{c}^{(N)}_k)^\top$
		\FOR{$i=1$ to $N$}
		\STATE $\Delta \eta_k^i \iid \normal(0,\frac{1}{\stepsize}R^{-1})$
		\STATE $\Delta{Y}^i_k = f({ Y}^i_k,\Delta \eta_k^i ) \Delta t +\frac{1}{2}
		M^{(N)}_k (c({ Y}^i_k)+\hat{c}^{(N)}_k) \Delta t$
		\STATE ${Y}^i_{k-1} = {Y}^i_k - \Delta{Y}^i_k$
		\ENDFOR	
		\STATE Calculate $n^{(N)}_{k-1} = N^{-1}\sum_{i=1}^N { Y}^i_{k-1}$ 
		\STATE Calculate	$ S^{(N)}_{k-1}= (N-1)^{-1}\sum_{i=1}^N ({ Y}^i_{k-1} -n^{(N)}_{k-1})({ Y}^i_{k-1} - n_{k-1}^{(N)})^\top$
		\STATE  $P^{(N)}_{k-1} = (\SN_{k-1})^{-1}$		
		\ENDFOR  	
	\end{algorithmic}
\end{algorithm}

\begin{algorithm}
    \caption{\textbf{[online]} EnKF algorithm to calculate optimal control for \eqref{eq:nonlinear_opt_control_problem}}
    \label{alg:EnKF}
    \begin{algorithmic}[1]
    	\REQUIRE Simulation time $T$, simulation step-size $\stepsize$, number of particles $N$, $\{P^{(N)}_k : k=0,1,2,\ldots,\frac{T}{\stepsize} \}$ from the offline algorithm~\ref{alg:P}, Hamiltonian function $ \mathcal{H}(x,y,\alpha) = y^T(a(x)+b(x)\alpha) +  \frac{1}{2}|c(x)|^2+\frac{1}{2}\alpha^\top R\alpha$, $\{e_i\}_{i=1}^{m}$ the standard basis of $\R^m$
    	\RETURN optimal control input  $\{  u_k^{(N)} \in \R^m : k=0,1,2,\ldots,\frac{T}{\stepsize} -1  \}$.
   	\STATE Define $T_F\coloneqq \frac{T}{\stepsize}$ 
%
        \FOR{$k=0$ to $T_F-1$}
                	\STATE Observe state of the system, denoted $x_k$
					\STATE Define $y_k = P_k x_k$          			
          			\FOR{$i=1$ to $m$}        
        	\STATE  $\langle u_k^{(N)}, e_i \rangle = \mathcal{H}(x_k,y_k,R^{-1}e_i) -  \mathcal{H}(x_k,y_k,0) - \frac{1}{2}(R^{-1})_{ii}$ 
        	\ENDFOR
          			\STATE Apply control $u^{(N)}_k$ to the true system           			     	
        \ENDFOR
    \end{algorithmic}
\end{algorithm}

\section{Details of Example \ref{sec:ex2}}

\subsection{Coupled mass spring damper system}
\label{app:smd-det}



This system is taken from~\cite{mohammadi_global_2019}.
The matrices $A$ and $B$ are as follows:
\begin{align*}
A = \begin{bmatrix}
0_{\dsp \times \dsp} & \id_{\dsp} \\ -\toep & -\toep
\end{bmatrix}, \quad B = \begin{bmatrix}
0_{\dsp \times \dsp} \\ \id_{\dsp}
\end{bmatrix}
\end{align*}
 where $\dsp=\frac{d}{2}$  is the number of masses and  $\toep \in
 \R^{\dsp \times \dsp}$ is a  Toeplitz matrix with $2$ on the main
 diagonal and $-1$ on the first sub-diagonal and first
 super-diagonal. Numerical values of parameters used in simulations are listed 
 in Table \ref{table:smd-mod}. 

\begingroup
\renewcommand*{\arraystretch}{1.4}
\begin{table}[h]
\centering
\caption{Model parameters for the coupled mass spring damper system}
\begin{tabular}{c c c} \hline
Parameter Name & Symbol & Numerical value \\ \hline
& \textbf{Model Parameters} & \\
\multirow{4}{*}{LQ parameters} & $C$ for $d=2$ & $\sqrt{5}\id_{d}$\\ 
& $C$ for $d > 2$ & $\id_{d}$\\ 
& $R$ & $\id_{\dsp}$\\ 
& $P_T$ & $\id_{d}$ \\ \hline
& \textbf{Simulation Parameters} & \\
Simulation time & $T$ & 10 \\ 
Step size & $\stepsize$ & 0.02 \\ \hline
\end{tabular}
\label{table:smd-mod}
\end{table}
\endgroup

\subsection{Comparison between EnKF and policy-gradient methods}

\label{app:comparison}

The hyper-parameters required to implement the algorithms of [M21], and [F18] algorithms are as follows.
The simulation time horizon $T=10$, and the step-size $\Delta t=0.01$ is the same for all of EnKF, [F18] and [M21].  The initial guess $K^0 = 0$, initial distribution $\mathcal{D}^0 = \normal(0,\id_d)$, and gradient descent step $\alpha = 0.0001$ for both [M21] and [F18]. The values of the other hyper parameters, namely the smoothing parameter $r$ and number of particles in gradient calculation $N_g$ are in Table \ref{table:params-gd-val}. The numerical results for $d=10$ are depicted in Figure \ref{fig:comp-to-lit}  and for $d=2,4$ in Figure \ref{fig:error-k}. Additionally, Figure \ref{fig:error-c} shows comparison for error in cost. While calculating cost, the system is initialised with a $\normal(0,0.1\id_d)$ distribution to keep the simulation setup as close to the setting of [M21] and [F18] as possible.

The simulations are  implemented in Python 3 on a Intel Xeon E3-1240 V2 3.40 Ghz CPU, and the {\tt process\textunderscore time() } function from the {\tt time} module is used to evaluate the execution time.

\begin{figure}[h]
		\centering
\vspace*{-0.1in}
	\includegraphics[width=0.8\columnwidth]{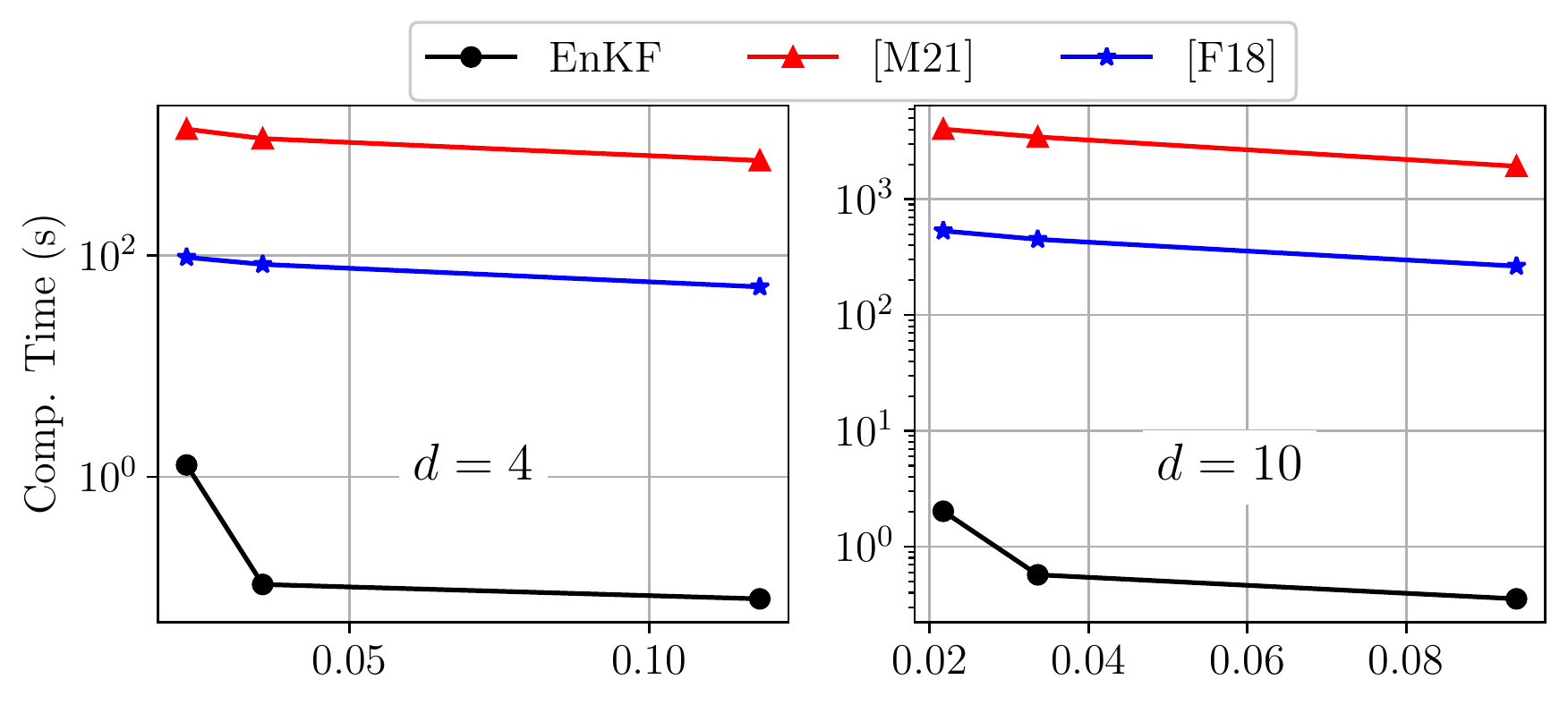}  
\vspace{-0.1in}
		\caption{Comparison of relative in error in gain}\label{fig:error-k}
	\end{figure}

\begin{figure}[h]
		\centering
\vspace*{-0.1in}
	\includegraphics[width=0.8\columnwidth]{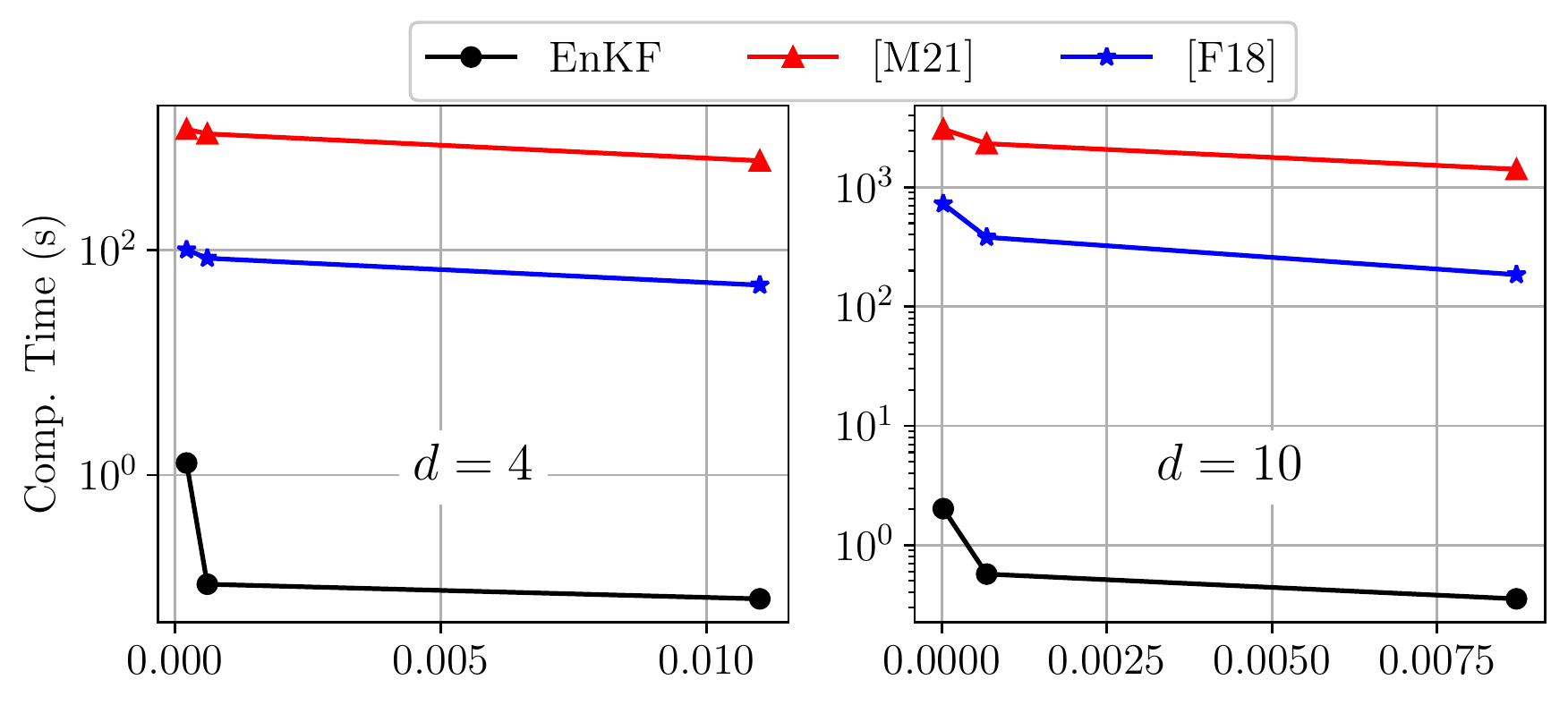}  
\vspace{-0.1in}
		\caption{Comparison of relative in error in cost}\label{fig:error-c}
\vspace{-0.2in}
	\end{figure}


\begingroup
\renewcommand*{\arraystretch}{1.4}
\begin{table}
\centering
\caption{Hyper-parameter values for policy gradient}
\label{table:params-gd-val}
\begin{tabular}{c|ccc|ccc} \hline
\multicolumn{1}{c}{Hyper-param.} & \multicolumn{3}{c}{[M21]} & \multicolumn{3}{c}{[F18]} \\ \hline
$d$ & 2 & $4$ & $10$ & $2$ & $4$ & $10$ \\
$r$ & $10^{-1}$ & $10^{-1}$ & $10^{-3}$ & $10^{-1}$ & $10^{-1}$ & $10^{-1}$ \\
$N_g$ & 2 & 4 & 10 & 2 & 4 & 10 \\ \hline
\end{tabular}
\end{table}
\endgroup

\section{Cart-pole system}
\label{app:cp-det}



The nonlinear model is taken from \cite[Chapter 3.2.1]{tedrake-notes}:
\begin{align*}
\dot{\theta} &= \omega \\
\dot{\omega} &= \frac{-F\cos(\theta) - ml\omega^2\cos(\theta)\sin(\theta) - (m+M)g\sin(\theta)}{l(M + m\sin^2(\theta))}  \\
\dot{x} &= v \\
\dot{v} &= \frac{F + m\sin(\theta)(l\omega^2 + g\cos(\theta))}{M + m\sin^2(\theta)} 
\end{align*}

For the specification of the LQ cost, we first linearize the system about the
desired inverted equilibrium $(\pi,0,0,0)$.  The associated $A$ and
$B$ matrices are as follows:
\begingroup
\renewcommand*{\arraystretch}{1.4}
\begin{align*}
A = \begin{bmatrix}
0 & 0 & 1 & 0 \\  
0 & 0 & 0 & 1 \\ 
\frac{(M+m)g}{Ml} & 0 & 0 & 0 \\
\frac{mg}{M} & 0 & 0 & 0
\end{bmatrix}, 
\quad 
B= \begin{bmatrix}
0 \\ \frac{1}{Ml} \\ 0 \\ \frac{1}{M} 
\end{bmatrix}
\end{align*}
Note these are used only to obtain the LQR solution (for comparison) but not needed to
implement the dual EnKF.  
The model parameters and the simulation parameters are are listed in Table \ref{table:ivp-param}.
\endgroup

%

\begingroup
\renewcommand*{\arraystretch}{1.4}
\begin{table}[h]
\centering
\caption{Parameters for the cart-pole system}
\begin{tabular}{c c c} \hline
Parameter name & Symbol & Numerical value \\ \hline
& \textbf{Model parameters} & \\
Mass of ball & $m$ & 0.08 \\ 
Mass of cart & $M$ & 1\\ 
Length of rod & $l$ & 0.7\\ 
Gravity & $g$ & 9.81\\  
Unstable equilibrium & $(\bar{\theta},\bar{x},\bar{\omega},\bar{v})$ & $(\pi,0,0,0)$\\ 
Initial condition & $(\theta(0),x(0),\omega(0),v(0))$ & $(1.25\pi,-0.1,0,0)$\\ \cline{1-3}
\multirow{3}{*}{LQ parameters}
 & $C$ & $\text{diag}([10,10,1,1])$\\ 
 & $R$ & 10\\ 
 & $P_T$ & $\id_4$\\ \hline
 & \textbf{Simulation parameters} & \\
 Simulation time & $T$ & 10 \\ 
Step size & $\stepsize$ & 0.0002 \\ \hline
\end{tabular}
\label{table:ivp-param}
\end{table}
\endgroup


  \bibliographystyle{elsarticle-num} 
  \bibliography{refs,literature}





\end{document}